\documentclass[11pt]{article}

\usepackage{amsmath}
\usepackage{amssymb}
\usepackage{amsthm}
\usepackage{mathabx}
\usepackage{graphicx} %
\usepackage{fullpage}
\usepackage{hyperref}
\usepackage{thmtools}
\usepackage[nameinlink,capitalise]{cleveref}
\usepackage{xcolor}
\usepackage{multicol}
\usepackage{url}
\usepackage[utf8]{inputenc}
\usepackage{enumitem}
\usepackage{tikz}

\newcommand{\zo}{\{0,1\}}

\newcommand{\im}{\operatorname{im}}
\newcommand{\CSP}{\operatorname{CSP}}

\newcommand{\Pol}{\operatorname{Pol}}

\newcommand{\rev}{\operatorname{rev}}

\newcommand{\cF}{\mathcal F}

\newcommand{\Q}{\mathbb{Q}}

\newcommand{\F}{\mathbb{F}}

\newcommand{\cC}{\mathcal C}
\newcommand{\cD}{\mathcal D}

\newcommand{\NRD}{\operatorname{NRD}}

\newcommand{\ncl}{\operatorname{ncl}}

\newcommand{\fG}{\mathfrak G}
\newcommand{\fM}{\mathfrak M}
\newcommand{\fN}{\mathfrak N}
\newcommand{\fU}{\mathfrak U}
\newcommand{\fV}{\mathfrak V}
\newcommand{\fF}{\mathfrak F}
\newcommand{\fR}{\mathfrak R}
\newcommand{\fL}{\mathfrak L}
\newcommand{\fQ}{\mathfrak Q}

\newcommand{\fW}{\mathfrak{W}}
\newcommand{\fe}{\mathfrak{e}}
\newcommand{\fg}{\mathfrak{g}}
\newcommand{\fw}{\mathfrak{w}}
\newcommand{\fu}{\mathfrak{u}}
\newcommand{\fm}{\mathfrak{m}}
\newcommand{\fn}{\mathfrak{n}}
\newcommand{\fv}{\mathfrak{v}}
\newcommand{\fq}{\mathfrak{q}}
\newcommand{\G}{\fG}

\newcommand{\T}{\mathbb{T}}

\newcommand{\cA}{\mathcal{A}}
\newcommand{\cB}{\mathcal{B}}
\newcommand{\cE}{\mathcal{E}}
\newcommand{\cV}{\mathcal{V}}
\newcommand{\cK}{\mathcal{K}}

\newcommand{\vol}{\operatorname{vol}}
\newcommand{\R}{\mathbb{R}}
\newcommand{\Z}{\mathbb{Z}}

\newcommand{\cT}{\mathcal{T}}

\newtheorem{theorem}{Theorem}
\numberwithin{theorem}{section}
\newtheorem{lemma}[theorem]{Lemma}

\newtheorem{proposition}[theorem]{Proposition}
\newtheorem{corollary}[theorem]{Corollary}

\newtheorem{conjecture}[theorem]{Conjecture}

\theoremstyle{definition}
\newtheorem{definition}[theorem]{Definition}
\newtheorem{remark}[theorem]{Remark}
\newtheorem{example}[theorem]{Example}
\newtheorem{fact}[theorem]{Fact}

\allowdisplaybreaks

\title{The Reach of Abelian Covers in Hypergraphs}
\author{Joshua Brakensiek\thanks{University of California, Berkeley. Supported in part by a Simons Investigator award of Venkatesan Guruswami, and NSF awards CCF-2211972 and DMS-2503280. Contact: \href{mailto:josh.brakensiek@berkeley.edu}{josh.brakensiek@berkeley.edu}} \and Venkatesan Guruswami\thanks{Simons Institute for the Theory of Computing and the University of California, Berkeley. Supported in part by a Simons Investigator award and NSF award CCF-2211972. Contact: \href{mailto:venkatg@berkeley.edu}{venkatg@berkeley.edu}} \and Aaron Putterman\thanks{Harvard University. Supported in part by a Jane Street Graduate Research Fellowship, the Simons Investigator awards of Madhu Sudan and Salil Vadhan, and AFOSR award FA9550-25-1-0112. Contact: \href{mailto:aputterman@g.harvard.edu}{aputterman@g.harvard.edu}}}
\date{}

\begin{document}

\maketitle

\begin{abstract}
In combinatorics and theoretical computer science, covers in hypergraphs are frequently studied to capture various forms of dependence between hyperedges. For example, even covers--which check if each vertex appears in an even number of hyperedges--have found much success recently in the study of locally decodable codes.
Inspired by a recently-emerging line of work on the non-redundancy of constraint satisfaction problems (CSPs), we introduce and study two novel families of covers of hypergraphs which are stricter than even covers: \emph{Abelian} covers and \emph{Catalan} covers. Abelian covers are similar to even covers, except that arithmetic is now done over the integers rather than modulo 2, allowing us to capture dependences over arbitrary Abelian groups. Catalan covers capture the behavior of non-Abelian groups by only allowing local cancellations in a sequence of hyperedges.

\smallskip We prove three main results about Abelian and Catalan covers. First, using tools from lattice theory, we show that any $r$-uniform hypergraph with $n$ vertices and $n \log(r)$ hyperedges has an Abelian cover. Second, using tools from algebraic topology, we show that in any $3$-uniform hypergraph, Abelian covers and Catalan covers are equivalent; thereby showing that Catalan covers emerge after $O(n)$ hyperedges in $3$-uniform hypergraphs. Finally, using the theory of nilpotent groups, we show that there exists a $4$-uniform hypergraph which has an Abelian cover but not a Catalan cover. Collectively, these results exactly characterize the reach that Abelian covers have in deducing dependences in hypergraphs. 
As our primary application, we show that any arity-$3$ CSP with an infinite-domain Mal'tsev extension has linear non-redundancy, answering an open question of Brakensiek et. al. (to appear, FOCS 2026). 
This implies near optimal streaming, sparsification, and kernelization algorithms for this family of CSPs. Previously, such a result was only known for the much simpler case of arity-$2$ CSPs. We see this result as further progress toward the goal of classifying exactly which CSPs have linear non-redundancy.
\end{abstract}

\pagenumbering{gobble}

\pagebreak

\tableofcontents 

\pagebreak

\pagenumbering{arabic}

\section{Introduction}

We introduce and study new classes of hypergraph cover problems, motivated by recent effort towards understanding the \emph{non-redundancy} of various constraint satisfaction problems (CSPs) \cite{chen2020BestCase, bessiere2020Chain, lagerkvist2020Sparsification, carbonnel2022Redundancy, KPS24, khanna2025efficient, brakensiek2025redundancy, brakensiek2025Richness, brakensiek2026classification, EGL26, SV26b, BGJLW26, sharma2026characterizing}.

\subsection{Covers in Hypergraphs}

Given a hypergraph $H = (V, E)$, a cover of the hypergraph captures basic ``dependence'' information among the hyperedges, and consequently has become a fundamental object of study both in combinatorics and theoretical computer science more broadly. As an example, the most basic and frequently studied cover of a hypergraph is the so-called \emph{even cover}. More formally, given an $r$-uniform hypergraph $H = (V, E)$, an even cover $C \subseteq E$ is a collection of hyperedges which touches every vertex an even number of times. A folklore result states that every hypergraph with more than $|V| = n$ hyperedges \emph{must} have an even cover. This is not hard to see either; indeed, we must only view each hyperedge $e \in E$ as its indicator vector $\mathbf{1}_e \in \F_2^n$. Given any $n+1$ such vectors, there must then be some linear dependence among them, i.e., some set of vectors such that $\sum_{e \in C} \mathbf{1}_e = 0$, which is exactly an even cover.

Finding even covers (particularly short ones) has since become an important subroutine for a variety of problems. For instance, such covers have led to recent progress on understanding the limits of \emph{locally decodable and correctable codes} \cite{alrabiah2023near, hsieh2025small, alrabiah2026near}, and are similarly important in the study of CSP refutations \cite{feige2006witnesses,guruswami2022algorithms,hsieh2023simple,hsieh2025small}. 
In this work, we introduce several new types of hypergraph covers, motivated by concrete applications to the study of non-redundancy in CSPs. The first such cover we introduce is the notion of an \emph{Abelian cover}. In this setting, we are again given a hypergraph $H = (V, E)$ and we say that a collection of hyperedges $C$ is an Abelian cover if there exists a hyperedge $e^* \in C$ along with integers $\alpha_e: e \in C - \{e^*\}$ such that
\[
\sum_{e \in C - \{e^* \}} \alpha_e \cdot \mathbf{1}_e =\mathbf{1}_{e^*},
\]
where the arithmetic is now done over \emph{integers}.
Intuitively, such a cover is a collection of hyperedges where an \emph{integer linear combination} of all but one of the hyperedges exactly yields the final hyperedge. Note that because the ``distinguished'' hyperedge $e^*$ only appears once, this cover is non-trivial for every modulus. %
In this way, Abelian covers are a much stricter 
notion than that of an even cover; simply taking the $\alpha_e$'s modulo $2$ would instead yield an even cover.
By the same reasoning, a single Abelian cover \emph{simultaneously} yields any ``$\text{mod } q$ cover'' for any choice of $q$ (and in fact, even yields covers for any Abelian group). 

\begin{figure}
    \centering
\begin{tikzpicture}[
    vertex/.style={
        circle, draw=black, thick, fill=white,
        minimum size=8mm, inner sep=0pt, font=\small
    },
    partbox/.style={
        draw=gray!45, fill=gray!8,
        rounded corners=5pt, line width=.8pt
    },
    hedge/.style={
        line width=1.6pt,
        rounded corners=5pt,
        opacity=.85
    }
]

\coordinate (v1) at (0,  1);
\coordinate (v2) at (0, -1);

\coordinate (v3) at (3,  1);
\coordinate (v4) at (3, -1);

\coordinate (v5) at (6,  1);
\coordinate (v6) at (6, -1);

\draw[partbox] (-0.6,-1.55) rectangle (0.6,1.55);
\draw[partbox] ( 2.4,-1.55) rectangle (3.6,1.55);
\draw[partbox] ( 5.4,-1.55) rectangle (6.6,1.55);

\draw[hedge, red!70!black] 
    (v1) -- (v3) -- (v6);

\draw[hedge, blue!70!black] 
    (v1) -- (v4) -- (v5);

\draw[hedge, green!50!black] 
    (v2) -- (v3) -- (v5);

\draw[hedge, orange!80!black] 
    (v2) -- (v4) -- (v6);

\node[font=\scriptsize, text=red!70!black]    at (1.2,  1.45) {$\{1,3,6\}$};
\node[font=\scriptsize, text=blue!70!black]   at (2.2,  -0.05) {$\{1,4,5\}$};
\node[font=\scriptsize, text=green!50!black]  at (4.2, 1.45) {$\{2,3,5\}$};
\node[font=\scriptsize, text=orange!80!black] at (4.2, -1.45) {$\{2,4,6\}$};

\node[vertex] at (v1) {$1$};
\node[vertex] at (v2) {$2$};

\node[vertex] at (v3) {$3$};
\node[vertex] at (v4) {$4$};

\node[vertex] at (v5) {$5$};
\node[vertex] at (v6) {$6$};

\end{tikzpicture}     \caption{An arity $3$ hypergraph which has an even cover, but no Abelian cover.}
    \label{fig:EvenButNotAbelian}
\end{figure}
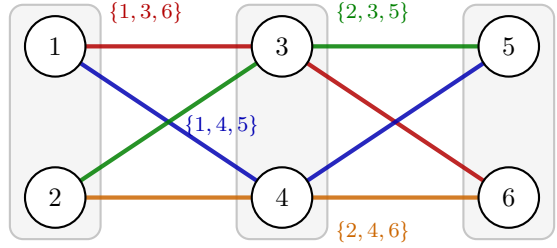

This notion of an Abelian cover is \emph{strictly stronger} than that of an even cover. This is immediate by looking, for instance, at a triangle of edges in an ordinary graph: such a collection of edges will constitute an even cover, but \emph{does not} constitute a cover modulo $3$ (and thus cannot constitute an Abelian cover either). Unfortunately, as we discuss more later, the triangle is not the most compelling separation between these two types of covers, as the most interesting setting of study for analyzing when covers emerge is in $r$-partite hypergraphs (or bipartite graphs in the setting $r = 2$), where such odd cycles are precluded. Nevertheless, even in these $r$-partite hypergraphs, even and Abelian covers can be simply separated.
As an example, we can consider the hypergraph in \cref{fig:EvenButNotAbelian}. 
It is not hard to see that this collection of hyperedges constitutes an even cover, as every vertex appears in two hyperedges. However, one can check that this hypergraph \emph{does not} admit any $\text{mod }3$ cover (and thus also does not admit any Abelian cover).

As mentioned above, we do know that even covers necessarily emerge in hypergraphs whenever the number of hyperedges exceeds $n$. It is natural to wonder whether the same is true even for \emph{Abelian covers}. This leads to the first central question that we study in this work:

\begin{center}
    \emph{How many hyperedges must a hypergraph have before an Abelian cover emerges?}
\end{center}

As our first result, we show that Abelian covers \emph{do emerge} after only a linear number of hyperedges:

\begin{theorem}\label{thm:AbelianCoverIntro}
    Let $H = (V, E)$ be an $r$-uniform hypergraph for any $r$. If $|E| \geq |V|\cdot (1 + \log_2(r)/2)$, then $H$ admits an Abelian cover. 
\end{theorem}

As we discuss further below in the introduction, this theorem already has implications in bounding the ``non-redundancy'' of CSPs arising from Abelian groups. 

From a combinatorial perspective, there is still more that one can ask for than an Abelian cover. Indeed, rather than viewing a cover as a collection of hyperedges, one can view a cover as a \emph{sequence} of hyperedges which generates some other hyperedge. Indeed, rather than measuring the contribution of each hyperedge using Abelian groups, one can instead define the aggregate of this sequence to be based on \emph{local cancellations} between hyperedges in the sequence. Can one still build covers when all the cancellations are forced to be local in the sequence of hyperedges?

While this question may seem unnatural, it turns out that this intuition is \emph{exactly} what appears in understanding the non-redundancy of certain CSPs which emerge from \emph{non-Abelian groups}---see \cite{brakensiek2025Richness} for a detailed discussion. We formalize this in a problem that we call the \emph{Catalan cover} problem. 

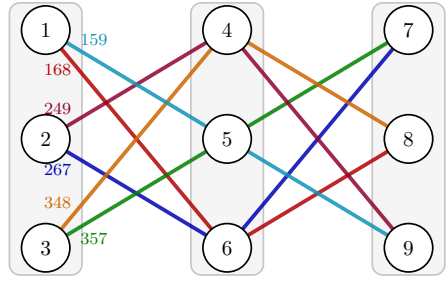
\begin{figure}[ht]
    \centering
    \scalebox{0.8}{%
\begin{tikzpicture}[
    vertex/.style={
        circle, draw=black, thick, fill=white,
        minimum size=8mm, inner sep=0pt, font=\small
    },
    partbox/.style={
        draw=gray!45, fill=gray!8,
        rounded corners=5pt, line width=.8pt
    },
    hedge/.style={
        line width=1.8pt,
        opacity=.85,
        rounded corners=5pt
    }
]

\coordinate (v1) at (0,  1.8);
\coordinate (v2) at (0,  0);
\coordinate (v3) at (0, -1.8);

\coordinate (v4) at (3,  1.8);
\coordinate (v5) at (3,  0);
\coordinate (v6) at (3, -1.8);

\coordinate (v7) at (6,  1.8);
\coordinate (v8) at (6,  0);
\coordinate (v9) at (6, -1.8);

\draw[partbox] (-0.6,-2.25) rectangle (0.6, 2.25);
\draw[partbox] ( 2.4,-2.25) rectangle (3.6, 2.25);
\draw[partbox] ( 5.4,-2.25) rectangle (6.6, 2.25);

\draw[hedge, red!70!black] 
    (v1) -- (v6) -- (v8);

\draw[hedge, blue!70!black] 
    (v2) -- (v6) -- (v7);

\draw[hedge, green!50!black] 
    (v3) -- (v5) -- (v7);

\draw[hedge, orange!80!black] 
    (v3) -- (v4) -- (v8);

\draw[hedge, purple!75!black] 
    (v2) -- (v4) -- (v9);

\draw[hedge, cyan!70!black] 
    (v1) -- (v5) -- (v9);

\node[font=\scriptsize, text=red!70!black]    at (0.2,1.15) {$168$};
\node[font=\scriptsize, text=blue!70!black]   at (0.2,-0.50) {$267$};
\node[font=\scriptsize, text=green!50!black]  at (0.8, -1.65) {$357$};
\node[font=\scriptsize, text=orange!80!black] at (0.2,-1.05) {$348$};
\node[font=\scriptsize, text=purple!75!black] at (0.2, 0.50) {$249$};
\node[font=\scriptsize, text=cyan!70!black]   at (0.8,1.65) {$159$};

\node[vertex] at (v1) {$1$};
\node[vertex] at (v2) {$2$};
\node[vertex] at (v3) {$3$};

\node[vertex] at (v4) {$4$};
\node[vertex] at (v5) {$5$};
\node[vertex] at (v6) {$6$};

\node[vertex] at (v7) {$7$};
\node[vertex] at (v8) {$8$};
\node[vertex] at (v9) {$9$};

\end{tikzpicture} %
}
    \caption{A hypergraph of arity $3$ on $9$ vertices which contains hyperedges $159, 168, 249, 267, 348$ and $357$.}
    \label{fig:Catalan}
\end{figure}

To properly describe this problem we first proceed by example, visiting the concrete hypergraph from \cref{fig:Catalan}. We claim that this hypergraph does contain a Catalan cover; namely the collection of hyperedges $C = \{(168), (267), (357), (348), (249) \}$, along with the ``distinguished hyperedge'' $e^* = (159)$.

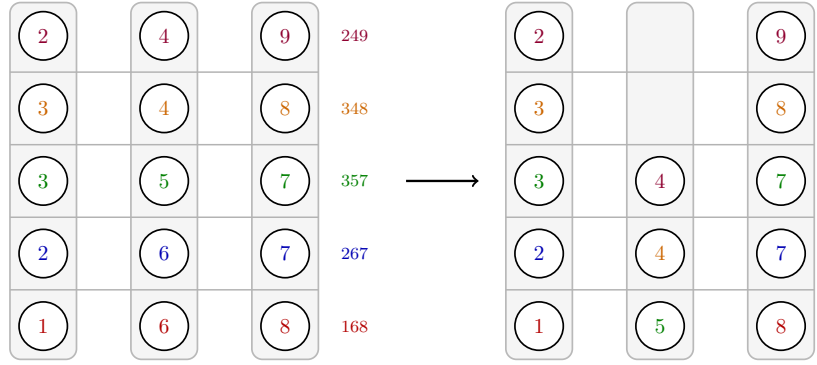
\begin{figure}
    \centering
    \scalebox{0.8}{%

\begin{tikzpicture}[
    digit/.style={
        circle,
        draw=black,
        thick,
        fill=white,
        minimum size=8mm,
        inner sep=0pt,
        font=\small
    },
    missingdigit/.style={
        circle,
        draw=none,
        fill=none,
        minimum size=8mm,
        inner sep=0pt
    },
    column/.style={
        draw=gray!55,
        fill=gray!8,
        rounded corners=5pt,
        line width=.8pt
    },
    rowlabel/.style={
        font=\scriptsize
    },
    separator/.style={
        draw=gray!60,
        line width=.7pt
    },
    arrow/.style={
        ->,
        thick,
        line width=1.1pt
    }
]

\begin{scope}[shift={(0,0)}]

\draw[column] (-0.55,-0.55) rectangle (0.55,5.35);
\draw[column] ( 1.45,-0.55) rectangle (2.55,5.35);
\draw[column] ( 3.45,-0.55) rectangle (4.55,5.35);

\draw[separator] (-0.55,0.6) -- (4.55,0.6);
\draw[separator] (-0.55,1.8) -- (4.55,1.8);
\draw[separator] (-0.55,3.0) -- (4.55,3.0);
\draw[separator] (-0.55,4.2) -- (4.55,4.2);

\node[digit, text=red!70!black] at (0,0) {$1$};
\node[digit, text=red!70!black] at (2,0) {$6$};
\node[digit, text=red!70!black] at (4,0) {$8$};
\node[rowlabel, text=red!70!black] at (5.15,0) {$168$};

\node[digit, text=blue!70!black] at (0,1.2) {$2$};
\node[digit, text=blue!70!black] at (2,1.2) {$6$};
\node[digit, text=blue!70!black] at (4,1.2) {$7$};
\node[rowlabel, text=blue!70!black] at (5.15,1.2) {$267$};

\node[digit, text=green!50!black] at (0,2.4) {$3$};
\node[digit, text=green!50!black] at (2,2.4) {$5$};
\node[digit, text=green!50!black] at (4,2.4) {$7$};
\node[rowlabel, text=green!50!black] at (5.15,2.4) {$357$};

\node[digit, text=orange!80!black] at (0,3.6) {$3$};
\node[digit, text=orange!80!black] at (2,3.6) {$4$};
\node[digit, text=orange!80!black] at (4,3.6) {$8$};
\node[rowlabel, text=orange!80!black] at (5.15,3.6) {$348$};

\node[digit, text=purple!75!black] at (0,4.8) {$2$};
\node[digit, text=purple!75!black] at (2,4.8) {$4$};
\node[digit, text=purple!75!black] at (4,4.8) {$9$};
\node[rowlabel, text=purple!75!black] at (5.15,4.8) {$249$};

\end{scope}

\draw[arrow] (6.0,2.4) -- (7.2,2.4);

\begin{scope}[shift={(8.2,0)}]

\draw[column] (-0.55,-0.55) rectangle (0.55,5.35);
\draw[column] ( 1.45,-0.55) rectangle (2.55,5.35);
\draw[column] ( 3.45,-0.55) rectangle (4.55,5.35);

\draw[separator] (-0.55,0.6) -- (4.55,0.6);
\draw[separator] (-0.55,1.8) -- (4.55,1.8);
\draw[separator] (-0.55,3.0) -- (4.55,3.0);
\draw[separator] (-0.55,4.2) -- (4.55,4.2);

\node[digit, text=red!70!black] at (0,0) {$1$};
\node[digit, text=green!50!black] at (2,0) {$5$};
\node[digit, text=red!70!black] at (4,0) {$8$};
\node[digit, text=blue!70!black] at (0,1.2) {$2$};
\node[digit, text=orange!80!black] at (2,1.2) {$4$};
\node[digit, text=blue!70!black] at (4,1.2) {$7$};
\node[digit, text=green!50!black] at (0,2.4) {$3$};
\node[digit, text=purple!75!black] at (2,2.4) {$4$};
\node[digit, text=green!50!black] at (4,2.4) {$7$};
\node[digit, text=orange!80!black] at (0,3.6) {$3$};
\node[digit, text=orange!80!black] at (4,3.6) {$8$};
\node[digit, text=purple!75!black] at (0,4.8) {$2$};
\node[digit, text=purple!75!black] at (4,4.8) {$9$};

\end{scope}

\end{tikzpicture}
}
    \caption{A stack of hyperedges which creates a Catalan cover.}
    \label{fig:columns}
\end{figure}

We visualize this collection of hyperedges in \cref{fig:columns}, seen as a ``stack'' where the hyperedges (written as a list of vertex labels) are placed horizontally, one on top of the other (thereby yielding $r$ columns, each of height equivalent to the number of hyperedges in the stack). Importantly, in this model, the only cancellation we allow is that whenever two vertically adjacent nodes have the same vertex label, they yield a cancellation. For instance, in the bottom two entries of the middle column, there are two $6$'s that are vertically adjacent, and thus can be canceled out. Performing this cancellation then simplifies the stack, yielding the right-hand side of \cref{fig:columns}. 
In fact, as one can see, one can continue doing cancellations in this manner until the entire stack reduces to $3$ columns of height $1$, with remaining entries $(159)$! Indeed, one can simply use the following set of operations.
\[
\begin{bmatrix}
    2 & 4& 9 \\
    3 & 4& 8 \\
    3 & 5 & 7 \\
    2 & 6 & 7 \\
    1 & 6 & 8 
\end{bmatrix} \rightarrow 
\begin{bmatrix}
    2 & & 9 \\
    3 & & 8 \\
    3 & 4 & 7 \\
    2 & 4 & 7 \\
    1 & 5 & 8 
\end{bmatrix} \rightarrow 
\begin{bmatrix}
    2 & & 9 \\
    3 & & 8 \\
    3 &  & 7 \\
    2 &  & 7 \\
    1 & 5 & 8 
\end{bmatrix} \rightarrow
\begin{bmatrix}
     & & 9 \\
     & & 8 \\
    2 &  & 7 \\
    2 &  & 7 \\
    1 & 5 & 8 
\end{bmatrix} \rightarrow
\begin{bmatrix}
     & & 9 \\
     & & 8 \\
     &  & 7 \\
     &  & 7 \\
    1 & 5 & 8 
\end{bmatrix} \rightarrow 
\begin{bmatrix}
     & &  \\
     & &  \\
     &  & 9 \\
     &  & 8 \\
    1 & 5 & 8 
\end{bmatrix} \rightarrow
\begin{bmatrix}
     & &  \\
     & &  \\
     &  &  \\
     &  &  \\
    1 & 5 &9 
\end{bmatrix}
\]

This is exactly the sense in which the collection of hyperedges $C$ constitutes a Catalan cover with the ``distinguished hyperedge'' $e^* = (159)$, as by merely using local cancellations, one can generate the remaining hyperedge $(159)$. This same procedure is well-defined on arbitrary hypergraphs of arity $r$:\footnote{Note that we assume here that the hypergraph is $r$-partite, thereby yielding a canonical ordering of the vertices within each hyperedge.} given a collection of hyperedges $C$ (where hyperedges may appear multiple times!\footnote{For example, consider the hypergraph on six vertices with edges $\{135,136,145,235,246\}.$}) and a distinguished hyperedge $e^* \notin C$, one creates a ``stack'' with $r$ columns, and stacks the hyperedges in these columns. One then performs the allowed adjacent cancellations, with the goal of generating the remaining hyperedge $e^*$.

It turns out that these Catalan covers are provably stricter 
than Abelian covers, in the sense that any collection of hyperedges which constitutes a Catalan cover \emph{also} constitutes an Abelian cover.\footnote{This follows from assigning $+1$ weight to any odd-indexed hyperedge in the stack, and $-1$ to any even-indexed hyperedge.}
Given our newly strengthened understanding of the thresholds at which Abelian covers emerge, this motivates the second key question from our work:

\begin{center}
    \emph{How many hyperedges must a hypergraph have before a Catalan cover emerges?}
\end{center}

Although we do not answer this question for all hypergraph arities, we do completely settle the question for hypergraphs of arity $\leq 3$.

\begin{theorem}\label{thm:CatalanArity3}
    Let $H = (V, E)$ be an $r$-partite $r$-uniform hypergraph for $r \leq 3$. If $|E| \geq C \cdot |V|$ for a sufficiently large constant $C > 0$, then $H$ admits a Catalan cover. In fact, \emph{any Abelian cover} is a Catalan cover.\footnote{When we say that an Abelian cover is a Catalan cover, we mean that for any set of hyperedges $C$ along with distinguished hyperedge $e^* \in C$ that constitutes an Abelian cover, the same choice of $C$ admits a distinguished hyperedge in the Catalan sense (although the actual implementation of the ``stack'' of hyperedges which witnesses the cancellations may be different).}
\end{theorem}

Although the $r$-partite assumption may seem significant, it is without loss of generality in downstream applications to CSPs (see \cref{subsec:CSP-NRD}). Moreover, \cref{thm:CatalanArity3} is already enough to conclude that Catalan covers emerge in \emph{arbitrary} (non-partite) arity $3$ hypergraphs after $O(n)$ hyperedges.\footnote{This follows from expressing any $3$-uniform hypergraph via its $3$-lift, and observing that any cover in this $3$-lift must also be a cover in the original hypergraph.} Note that in the setting when $r = 2$, the proof is simple: a graph has an Abelian cover if and only if it has a cycle of even length. Take all but one of the edges of the cycle and place them onto the stack in consecutive order---after the cancellations, one is left with the final edge of the cycle. %
The primary technical contribution in \cref{thm:CatalanArity3} is a resolution of the $r = 3$ case by connecting Catalan covers to deep topological phenomena emerging from hypergraph structures. As we discuss in more detail in the technical overview (\cref{subsec:topology-intro}), we show the existence of Catalan covers when $r=3$ is closely tied to the \emph{fundamental group} of a topological space derived from the hypergraph $H$. Using a theorem of Hurewicz (e.g., \cite{Hatcher02}), we relate the fundamental group of this topological space to its Abelian \emph{homology} group, from which we can conclude that Abelian and Catalan covers are equivalent.

Nevertheless, \cref{thm:CatalanArity3} is suggestive: is it possible that for higher arities too, it is still the case that Abelian covers yield Catalan covers? For this question, we provide a strong negative result. Already at arity $4$, there are exact separations between Abelian and Catalan covers:

\begin{theorem}\label{thm:4uniformIntro}
    There exists a $4$-partite $4$-uniform hypergraph $H = (V, E)$ which admits an Abelian cover \emph{but not} a Catalan cover.
\end{theorem}

We note that the recent work \cite{brakensiek2025Richness} implicitly gave a $6$-uniform hypergraph with this property. We discuss this hypergraph in \cref{subsec:4-unif} of the technical overview.

\subsection{Connections to CSP Non-Redundancy}\label{subsec:CSP-NRD}

As mentioned above, the hypergraph cover problems we study are motivated by recent developments in the study of CSP \emph{non-redundancy}. In this setting, one is given an arity $r$ \emph{relation} $R \subseteq D^r$ (sometimes called a predicate), and a universe of $n$ variables $x_1, \dots x_n \in D$. The goal is to understand the maximum size of a \emph{non-redundant} CSP instance: this is a collection of sets (sometimes called \emph{constraints}) $S_1, \dots S_m \in \binom{[n]}{r}$ such that for each $i \in [m]$, there exists an assignment $x \in D^n$ such that $x|_{S_i} \notin R$, but $x|_{S_{j}} \in R$ for every $j \neq i$. Intuitively, this asks for a collection of constraints such that each constraint is \emph{uniquely unsatisfiable}; 
i.e., that there is an assignment which satisfies all other constraints $S_j: j \neq i$, but does not satisfy the constraint $S_i$. Equivalently, dropping any constraint alters (increases) the set of satisfying assignments of the instance.  We use $\NRD(R, n)$ to denote the maximum size of a non-redundant CSP instance on $n$ variables for the relation $R$.

Recent work has identified $\NRD(R, n)$ as a central parameter in the study of CSPs. For instance, the work of Brakensiek and Guruswami \cite{brakensiek2025redundancy} showed that the non-redundancy of a CSP exactly governs its \emph{sparsifiability}, thereby unifying a line of work on CSP sparsification \cite{KK15, FK17, BZ20,KPS24, khanna2025efficient}. Non-redundancy has also been identified as a critical parameter in CSP \emph{kernelization} \cite{lagerkvist2017kernelization, lagerkvist2020Sparsification} and the query complexity of learning CSPs \cite{bessiere2020Chain}. Even more recently, non-redundancy was also identified as the key parameter in the streaming complexity of deciding satisfiability of CSPs \cite{sharma2026characterizing}.

Motivated by this, a central direction which has emerged in the study of CSP non-redundancy is to understand for which relations $R \subseteq D^r$ have the smallest possible non-redundancy of $O(n)$ (see, for instance \cite{chen2020BestCase, bessiere2020Chain, lagerkvist2020Sparsification, carbonnel2022Redundancy, KPS24, khanna2025efficient, brakensiek2025redundancy, brakensiek2025Richness, brakensiek2026classification, EGL26, SV26b, BGJLW26}). Nevertheless, our understanding of which relations admit linear non-redundancy is not complete. Until now, the most general characterization relies on the notion of \emph{extensions}:\footnote{Many prior works (especially \cite{bessiere2020Chain}) use the term ``embedding'' rather than ``extension.'' We choose to use the latter term as the former can lead to confusion with the notion of ``Abelian embeddings'' which has recently come to prominence in the study of CSP approximability and additive combinatorics~\cite{BKM26,M26}. The first use of an ``extension'' of a CSP appears in \cite{lagerkvist2020Sparsification}.}

\begin{definition}
    Given two structures $(D, R)$ and $(E, S)$, we say that $S$ is an extension of $R$ if there exists an injective homomorphism $\sigma: D \rightarrow E$ from $R$ to $S$ such that $\sigma(R) = \sigma(D)^r \cap S$. 
\end{definition}

 Note that, as observed in prior works ~\cite{bessiere2020Chain,lagerkvist2020Sparsification}, if an extension exists, then it is necessarily the case that $\NRD(R, n) \leq \NRD(S, n)$.

 Beyond simply finding extensions, 
 deriving characterizations of non-redundancy relies on finding extensions that admit special \emph{group} structure. The most general such structure is the notion of a \emph{Mal'tsev polymorphism}. Here, one desires a map $\varphi: E^3 \rightarrow E$ such that for any $x, y \in E$, we have $\varphi(x, y, y) = \varphi(y, y, x) = x$, and further, for any $t_1, t_2, t_3 \in S$, we also have that $\varphi(t_1, t_2, t_3) \in S$. One way to realize such a relationship is to enforce that $(E, \cdot)$ is a group, and the set $S$ is a coset of $E$. Then the typical group operation defined as $\varphi(x, y, z) = x y^{-1} z$ stays within the coset, hence satisfying the above conditions. 

 A fundamental result in the study of CSPs shows that any relation $R \subseteq D^r$ which admits a Mal'tsev extension must inherently have bounded non-redundancy:

 \begin{fact}[Proposition 9~\cite{bessiere2020Chain}, see also \cite{BD06,lagerkvist2020Sparsification}]\label{fact:malt}
     Let $R \subseteq D^r$ be a relation which admits an extension into a structure $(E, S)$ which admits a Mal'tsev polymorphism. Then, $\NRD(R, n) \leq |E|^2 \cdot n + 1$.
 \end{fact}

 Importantly however, this bound on the non-redundancy \emph{scales} with the domain size of the resulting extension $E$. Unfortunately, this quantity may even be infinite! Because of this, \cite{brakensiek2025Richness} explicitly poses a conjecture to \emph{bound} the size of these extension when they exist:

 \begin{conjecture}[Conjecture 8.1 of \cite{brakensiek2025Richness}]\label{conj:rich}
     Any relation $R \subseteq D^r$ over a finite domain with an infinite Mal’tsev extension also has
a finite Mal’tsev extension. In particular, any relation with an infinite Mal’tsev extension has
linear non-redundancy.
 \end{conjecture}

It is important to highlight that \emph{all known relations} with proven linear non-redundancy satisfy the conditions of \cref{conj:rich}. In this sense, the setting of \cref{conj:rich} is even a candidate for being an \emph{exact} characterization of \emph{all} relations which have linear non-redundancy. Unfortunately, both directions of this characterization are unproven; i.e., that all relations satisfying the conditions of \cref{conj:rich} have linear non-redundancy, and that all relations with linear non-redundancy admit infinite Mal’tsev extensions. Our work can thus be seen as progress towards proving this first direction.

The authors of \cite{brakensiek2025Richness} prove this conjecture in the special case for which $D = \{0,1\}$. As we highlight in this work, the problem of showing that all Mal’tsev extensions have linear non-redundancy can be reduced to showing that Catalan covers of a hypergraph emerge after only $O(n)$ hyperedges. In particular, a rather self-contained combinatorial question on hypergraphs would imply \cref{conj:rich} could have a rather profound impact on the universal-algebraic structure of CSPs. %
We formalize this conjecture as follows.

\begin{conjecture}\label{conj:linear}
For all $r \in \mathbb N$, there exists a constant $C_r > 0$ such that any $r$-partite $r$-uniform hypergraph on $n$ vertices with at least $C_r n$ hyperedges has a Catalan cover.
\end{conjecture}

As a consequence of \cref{thm:CatalanArity3}, we now know that \cref{conj:linear} is true when $r \le 3$.  Thus, we also prove \cref{conj:rich} is true when $r \leq 3$. Furthermore, as an immediate corollary of \cref{thm:CatalanArity3}, we obtain the following:

\begin{corollary}\label{cor:arity3Maltsev}
    Let $R \subseteq D^3$ be a relation which admits a (potentially infinite) Mal'tsev extension. Then, $R$ has a (finite) Abelian extension and thus, $\NRD(R, n) = O(n)$, where the hidden constant does not depend on $D$.
\end{corollary}

The proof of this aforementioned corollary, even when given our bounds (\cref{thm:CatalanArity3}) on when Catalan covers emerge crucially uses a nontrivial fact of \cite{brakensiek2025Richness} that ``Catalan polymorphisms'' (closely related to Catalan covers) show that the non-redundancy of a Mal'tsev CSP implies the lack of a Catalan cover in the instance's underlying hypergraph.

\subsection{Technical Overview}

To illustrate the technical novelty of our main results, we now highlight the main technical insights going into \cref{thm:AbelianCoverIntro}, \cref{thm:CatalanArity3}, and \cref{thm:4uniformIntro}.

\subsubsection{Sketch of Proof of \texorpdfstring{\cref{thm:AbelianCoverIntro}}{Theorem 1.1}}

The proof of \cref{thm:AbelianCoverIntro} proceeds by recasting the existence of Abelian covers to a suitable strutural property of \emph{integer lattices}. More precisely, given vectors $v_1, \hdots, v_k \in \mathbb Z^n$, one can define a suitable lattice generated by these vectors as follows
\[
    \Lambda(v_1, \hdots, v_k) := \left\{\sum_{i=1}^k \alpha_i v_i \,\middle\vert\, \forall i \in [k], \alpha_i \in \mathbb Z \right\}
\]
Consider any hypergraph $H = (V, E)$ with $E = \{e_1, \hdots, e_m\}$. For any $i \in [m]$, define the lattice
\[
    \Lambda_i := \Lambda(\mathbf{1}_{e_1}, \mathbf{1}_{e_2}, \hdots, \mathbf{1}_{e_i}).
\]
Clearly $\Lambda_i \subseteq \Lambda_{i+1}$ for all $i \in [m-1]$. However, if it is the case that $\Lambda_i = \Lambda_{i+1}$ for some $i \in [m-1]$, then $\mathbf{1}_{e_{i+1}} \in \Lambda_i$. In particular, this implies there is some $E' \subseteq \{e_1, \hdots, e_{i+1}\}$ which forms an Abelian cover. Thus, if $H$ lacks an Abelian cover we have the following strict containment of lattices
\[
    \Lambda_1 \subsetneq \Lambda_2 \subsetneq \cdots \subsetneq \Lambda_{m}.
\]
Recall our goal is to show that $m = |E| = O(|V|)$. To do this, we closely study the \emph{quotient} of consecutive lattices in this chain. For each $i \in [m-1]$, there are two possible cases. 

The first case is the quotient $\Lambda_{i+1} / \Lambda_i$ is infinite. Since the dimension of $\mathbb Z^{|V|}$ is $|V|$, there are at most $|V|-1$ indices where this can occur.  The second case is that the quotient $\Lambda_{i+1} / \Lambda_i$ is finite, whose cardinality (index) we denote by $[\Lambda_{i+1} : \Lambda_i] \ge 2$. To handle this case, we keep track of a quantity $\vol \Lambda_i$ denoting the (Euclidean) volume of the smallest nontrivial parallelepiped using vectors of $\Lambda_i$. Crucially, when $[\Lambda_{i+1} : \Lambda_i]$ is finite, we have that
\[
    \frac{\vol \Lambda_{i+1}}{\vol \Lambda_i} = \frac{1}{[\Lambda_{i+1} : \Lambda_i]} \le \frac{1}{2}.
\]
Furthermore, since each vector edge $\mathbf{1}_{e_i}$ has $\ell_2$ distance of $\sqrt{r}$, we can show that
\[
\frac{\vol \Lambda_{i+1}}{\vol \Lambda_i} \le \sqrt{r},
\]
when $\Lambda_{i+1} / \Lambda_i$ is infinite. In combination with the fact that $\vol \Lambda_1 = \sqrt{r}$ and $\vol \Lambda_m \ge 1$, we can deduce that
\begin{align*}
    \frac{1}{\sqrt{r}} \le \frac{\vol \Lambda_m}{\vol \Lambda_1} = \prod_{i \in [m-1]} \frac{\vol \Lambda_{i+1}}{\vol \Lambda_1} &\le \prod_{\substack{i \in [m-1] \\ [\Lambda_{i+1} : \Lambda_i] < \infty}}\frac{1}{2} \cdot  \prod_{\substack{i \in [m-1] \\ [\Lambda_{i+1} : \Lambda_i] = \infty}} \sqrt{r}\\
    &\le \frac{1}{2^m} \cdot \sqrt{r}^{|V|-1},
\end{align*}

from which we can deduce that $m = O(|V|\log r)$, as desired.

\subsubsection{Sketch of Proof of \texorpdfstring{\cref{thm:CatalanArity3}}{Theorem 1.2}}\label{subsec:topology-intro}

The proof of \cref{thm:CatalanArity3} is rather technical, requiring a number of nontrivial ideas from group theory which are inspired from topology. At a high level, we adopt the following strategy. First, we model Catalan covers as a group which we call the \emph{Catalan group}. Second, we build a combinatorial \emph{topological space} $\mathbb T$ related to our tripartite hypergraph $H$. Using our Catalan group, we show that the \emph{fundamental group} of $\mathbb T$ captures the existence of Catalan covers, and the \emph{first homology group} of $\mathbb T$ captures the existence of Abelian covers. Finally, by an invocation of the Hurewicz theorem, we closely relate the fundamental group and first homology group of $\mathbb T$, showing the equivalence of Catalan and Abelian covers in tripartite $3$-uniform hypergraphs. Although these topological connections inspired our argument, the actual proof in \cref{sec:Catalan,sec:cover-3} is entirely self-contained.

As a first step, we model the notion of a Catalan cover as a group, an idea implicit in \cite{brakensiek2025Richness}. Given a vertex set $V$, we introduce a variable $\fg_v$ for each vertex $v \in V$, with the only promise being that $\fg_v^2 = 1$. A word is then simply a product of these variables; for instance if $V = \{1,2,3,4\}$, then $\fg_1 \fg_3 \fg_3 \fg_1 \fg_4 $ is a word, and it can be simplified to $\fg_4$, as 
\[
\fg_1 \fg_3 \fg_3 \fg_1 \fg_4 = \fg_1 \fg_1 \fg_4 = \fg_4.
\]
Let $\fG$ be the group generated by these words. Given a hyperedge $e$, we let $v_i(e)$ denote the $i$th vertex in $e$. Given an ordered collection of hyperedges $C$ (note that hyperedges are permitted to appear multiple times), we let $w_1(C) = \prod_{e \in C} \fg_{v_1(e)}$ denote the word which is formed as a result of taking the variables which correspond to each of the first vertices in each hyperedge $e \in C$, and similarly for $i \in [r]$, define $w_i(C) = \prod_{e \in C} \fg_{v_i(e)}$.

With this group-theoretic machinery, we can prove that an ordered collection of hyperedges $C$ along with a hyperedge $e^* \notin C$ forms a Catalan cover if for every $i \in [r]$ (where $r=3$ in \cref{thm:CatalanArity3}), 
\begin{align}
w_i(C) = \fg_{v_i(e^*)}.\label{eq:w}
\end{align}
In other words, the ``stack'' definition of Catalan covers is equivalent to the group-theoretic version.

Recall that we assume the hypergraphs under consideration are tripartite. Thus, we can view $E \subseteq V_1 \times V_2 \times V_3$ for some tripartition $V_1 \cup V_2 \cup V_3 = V$.  Given an edge $e \in E$, we can then define an edge element $\fg_e = (\fg_{v_1(e)}, \fg_{v_2(e)}, \fg_{v_3(e)}) \in \fG_1 \times \fG_2 \times \fG_3$,\footnote{In the actual proof, instead of looking at $\fG_1 \times \fG_2 \times \fG_3$ we instead look at just $\fG$ but impose ``commutativity conditions'' between vertices in different parts of the hypergraph. See \cref{sec:Catalan} for more details.} where $\fG_i$ is the subgroup of $\fG$ generated by $\{\fg_v \mid v \in V_i\}$. Given our edge set $E$, we can then let $\langle E\rangle_{\fG}$ denote the group generated by $\fg_e$ for $e \in E$. Checking if $(V, E)$ has a Catalan cover by \cref{eq:w} is then equivalent to checking if there is some $e^* \in E$ for which $\fg_{e^*} \in \langle E \setminus \{e^*\}\rangle_{\fG}$, or equivalently $\langle E \setminus \{e^*\}\rangle_{\fG} = \langle E\rangle_{\fG}$.

We also introduce the notation $\langle E\rangle_{\mathbb Z}$ for the lattice $\Lambda(e : e \in E)$ from the proof of \cref{thm:AbelianCoverIntro}. With the facts we already have laid out, proving \cref{thm:CatalanArity3} is equivalent to the following statement.

\begin{fact}\label{fact:AbelianCatalanIntro}
For all $E' \subseteq E \subseteq V^3$ if $\langle E'\rangle_{\mathbb Z} = \langle E\rangle_{\mathbb Z}$ then $\langle E'\rangle_{\fG} = \langle E\rangle_{\fG}$.
\end{fact}

Rather surprisingly the proof of \cref{fact:AbelianCatalanIntro} is highly topological. Recall from the discussion immediately after \cref{thm:CatalanArity3} that bipartite $2$-uniform Catalan covers and Abelian covers are easy to describe: they always arise from even-length cycles. By ``dropping'' the third part $V_3$ from $E$, we get a bipartite graph $E_{1,2} := \{(v_1(e), v_2(e)) \mid e \in E\}$ whose even cycles give strong hints concerning the structure of the Abelian and Catalan covers of $E$ itself.

To formalize this intuition, we build the following bipartite \emph{auxiliary graph} which we also view as a topological space. We build equivalence relations $\sim_1$ and $\sim_2$ on the third vertex set $V_3$ as follows. For $i \in [2]$, we say that $u,v \in V_3$ satisfy $u \widehat{\sim}_i v$ if there exist edges $e, f \in E$ such that (1) $v_i(e) = v_i(f)$, (2) $v_3(e) = u$, and (3) $v_3(f) = v$. From the perspective of the Catalan group, this relation captures the fact that one can ``translate'' between tuples which have $u$ and $v$ in their third position \emph{without} altering the word in their first position.
Such a relation $\widehat{\sim}_i$ might not be transitive, so we let $\sim_i$ denote the transitive closure of this relation. Our auxiliary graph $\mathbb T$ is built by having the ``left'' vertices be the equivalence classes of $\sim_1$ and its ``right'' vertices be the equivalence classes of $\sim_2$. For each $v \in V_3$, we have an edge (with label $v$) connecting its $\sim_1$ equivalence class to its $\sim_2$ equivalence class. %
The definition of $\mathbb T$ depends greatly on the choice of $E$, but crucially, we show that if $\langle E'\rangle_{\mathbb Z} = \langle E\rangle_{\mathbb Z}$, then we get the same $\mathbb T$ for both $E'$ and $E$. Thus, we can say $\mathbb T$ is \emph{the} topological space associated with our problem. %

Using $\mathbb T$, we can define a crucial subgroup of $\fG_3$, which we call the \emph{walk subgroup} $\fW \subseteq \fG_3$, which can be considered to be a discrete analogue of the topological \emph{fundamental group} (i.e., first homotopy group) of $\mathbb T$. More precisely, pick an arbitrary equivalence class $A_0$ of $\sim_1$, which we call the \emph{base point} of $\mathbb T$, now consider all walks on the graph $\mathbb T$ which start and stop at $A_0$. Since each edge of $\mathbb T$ is labeled by a vertex $v \in V_3$, we can uniquely determine a walk by its sequence $v_1, \hdots, v_k \in V_3$ of edges used. From this, we can define a corresponding \emph{cycle word} $\fg_{v_1} \cdots \fg_{v_k} \in \fG_3$. We let $\fW \subseteq \fG_3$ be the group generated by these cycle words.

With some simple group theory, we can show that $\langle E'\rangle_{\fG} \neq \langle E\rangle_{\fG}$ if and only if there is a cycle word $\fw \in \fW$ such that $(1,1,\fw) \in \langle E\rangle_{\fG} \setminus \langle E'\rangle_{\fG}$, where $1$ is the identity element of $\fG$.  As such our goal is now to show that if $\langle E'\rangle_{\mathbb Z} = \langle E\rangle_{\mathbb Z}$ then such a cycle word $\fw$ cannot exist. We analyze this by constructing an auxiliary group which we call the \emph{label group}.
\[
\langle E\rangle_{\fL} := \{\fu \in \fG \mid (1, 1, \fu) \in \langle E\rangle_{\fG}\},
\]
Intuitively, the way to think about this is that words in $\langle E\rangle_{\fL}$ are produced by finding \emph{cycles} in the graph produced by restricting each hyperedge of $E$ only to its first two variables. When one translates these cycles \emph{back} into the hypergraph $H$, the words in the first and second position \emph{still} cancel out, leaving a tuple of the form $(1, 1, w) $.
With some basic group theory, we can show that in order to prove \cref{fact:AbelianCatalanIntro}, it suffices to prove that if $\langle E'\rangle_{\mathbb Z} = \langle E\rangle_{\mathbb Z}$ then $\langle E'\rangle_{\fL} = \langle E\rangle_{\fL}$.  As a result, the statement $(1,1,\fw) \in \langle E\rangle_{\fG} \setminus \langle E'\rangle_{\fG}$ translates to the cosets $\fw \cdot \langle E'\rangle_{\fL}$ and $\langle E'\rangle_{\fL}$ being distinct.

Crucially, under a mild assumption on the structure of $E$, we have that $\langle E'\rangle_{\fL}$ and $\langle E\rangle_{\fL}$ are \emph{normal} subgroups of $\fG_3$. That is, the cosets of $\langle E'\rangle_{\fL}$ form a quotient group $\fG_3 / \langle E'\rangle_{\fL}$. The most pivotal step in the proof is realizing that this quotient group is ``mostly'' Abelian in the sense that the cosets correspond to any two cycle words commute!

\begin{fact}[see \cref{prop:cycle-commute}]\label{fact:commute}
For any cycles words $\fw_1, \fw_2 \in \fW$, we have that $\fw_1\fw_2 \cdot \langle E'\rangle_{\fL} = \fw_2\fw_1 \cdot \langle E'\rangle_{\fL}$.
\end{fact}

That is, quotienting by $\langle E'\rangle_{\fL}$ effectively \emph{Abelianizes} the fundamental group of $\T$! By the Hurewicz theorem (e.g., \cite{Hatcher02,sunada2012topological}), this Abelianized fundamental group is isomorphic to the \emph{first homology group} of $\T$, essentially corresponding to $\mathbb Z$-linear combinations of cycles of $\T$, which intuitvely captures the structure of Abelian covers of our hypergraph $E$. We note that a natural arity-4 generalization of \cref{fact:commute} does not hold, highlighting why \cref{thm:CatalanArity3} cannot extend to arity-$4$ hypergraphs.

At this point, we can relatively quickly complete a proof by contradiction. Assume there is some cycle word $\fw \in \fW$ for which $\fw \in \langle E\rangle_{\fL} \setminus \langle E'\rangle_{\fL}$. Let $\fw \cdot \langle E'\rangle_{\fL}$ be the corresponding cycle word coset. Using the fact that $\langle E'\rangle_{\Z} = \langle E\rangle_{\Z}$, we can break up $\fw \cdot \langle E'\rangle_{\fL}$ into a sum of cycle words cosets $\left(\sum_{i=1}^k \alpha_i \fw_i\right) \cdot \langle E'\rangle_{\fL}$ (with $\alpha_i \in \mathbb Z$ and $\fw_i \in \fW$) such that each $\fw_i$ is actually contained in $\langle E'\rangle_{\fL}$. In that case, $\left(\sum_{i=1}^k \alpha_i \fw_i\right) \cdot \langle E'\rangle_{\fL} = \langle E'\rangle_{\fL}$, so $\fw \cdot \langle E'\rangle_{\fL} = \langle E'\rangle_{\fL},$ contradicting the fact that $\fw \in \langle E\rangle_{\fL} \setminus \langle E'\rangle_{\fL}$. Therefore, \cref{fact:AbelianCatalanIntro} and thus \cref{thm:CatalanArity3} follows. We emphasize again that the proof of \cref{thm:CatalanArity3} is entirely self-contained with all the topological machinery proved by hand for completeness.

\subsubsection{Sketch of Proof of \texorpdfstring{\cref{thm:4uniformIntro}}{Theorem 1.3}}\label{subsec:4-unif}

Before we explain the proof of \cref{thm:4uniformIntro}, we first briefly describe the arity-$6$ hypergraph implicitly\footnote{More precisely, they construct a relation $R \subseteq \{0,1,2\}^6$ which has a Mal'tsev extension but not an Abelian extension. We interpret their construction as a $6$-uniform hypergraph with an Abelian cover but not a Catalan cover.} constructed by \cite[Theorem 6.18]{brakensiek2025Richness} which has an Abelian cover but not a Catalan cover. The hypergraph is on $18$ vertices $V := \{a_1, \hdots, a_6, b_1, \hdots, b_6, c_1, \hdots, c_6\}$ and consists of $6$ hyperedges
\begin{align*}
    E := \{&e_1 := (a_1,a_2,c_3,a_4,c_5,a_6),\ 
e_2 := (a_1,a_2,a_3,b_4,a_5,b_6),\ 
e_3 := (b_1,c_2,a_3,b_4,b_5,c_6),\\
&e_4 := (b_1,b_2,b_3,a_4,b_5,a_6),\ 
e_5 := (c_1,b_2,b_3,c_4,a_5,b_6),\ 
e_6 := (c_1,c_2,c_3,c_4,c_6,c_6)\}
\end{align*}
Note that $E$ is an Abelian cover due to the relation
\[
    \mathbf{1}_{e_1} - \mathbf{1}_{e_2} + \mathbf{1}_{e_3} - \mathbf{1}_{e_4} + \mathbf{1}_{e_5} = \mathbf{1}_{e_6}
\]
However, the order the edges of $E$ are written in hints as to why $E$ is \emph{not} a Catalan cover. Look at the stack $(e_1, e_2, e_3, e_4, e_5)$, the first 5 coordinates will have sufficient cancellation to coincide with $e_6$, but the last one does not due to ``long range cancellations.'' In fact this is a necessary feature--since the third Catalan number is 5, it is impossible for any permutation of $e_1, \hdots, e_6$ to have every coordinate coorrespond to a balanced sequence of parentheses. However, edges can appear multiple times to witness a Catalan cover, so this is not an exhaustive proof. To get around this, the authors of \cite{brakensiek2025Richness} build an important \emph{group extension} of the relation. More precisely, they construct a group $G$ and an extension map $\psi : V \to G$ such that each edge $e_i$ becomes a group element
\[
    \psi(e_i) := (\psi(e_{i,1}), \hdots, \psi(e_{i,6})) \in G^6.
\]
What makes such a group extension useful is that if there does exist a stack such that $[e_{i_1}, \hdots, e_{i_k}] = e_{i_{k+1}}$, then we have the group relation
\begin{align}
    \psi(e_{i_1}) \psi(e_{i_2})^{-1} \psi(e_{i_3}) \cdots \psi(e_{i_k}) = \psi(e_{i_{k+1}}).\label{group:eq}
\end{align}
Thus, it suffices to show that \cref{group:eq} can never occur if $i_{k+1} \not\in \{i_1, \hdots, i_k\}$. To do this, it suffices\footnote{Technically, to rule out a Catalan cover, we need to show that the coset generated by $\{\psi(e_1), \hdots, \psi(e_6)\} \setminus \psi(e_i)$ fails to contain $\psi(e_i)$ for all $i \in [6]$. However, by symmetry in $E$ and the group extension, it suffices to consider the case $i = 6$.} to show that the coset of $G^6$ generated by $\psi(e_1), \hdots, \psi(e_5)$ does not contain $\psi(e_6)$. Carrying out such a proof greatly depends on the group $G$ chosen, but \cite{brakensiek2025Richness} chose $G$ to be \emph{Pauli group} of $1$-qubit unitaries from which special structural properties of the Pauli group are used to establish the non-containment, particularly that the generators of the Pauli group have a certain ``anticommutative'' property.

To prove \cref{thm:4uniformIntro}, we need a few new techniques. First, the edges $E$ chosen by \cite{brakensiek2025Richness} have the shortcoming that any ``restriction'' to $5$-uniform edges fails to separate Abelian covers from Catalan covers, so we need an entirely new construction. Using a computer search, we found the following candidate 4-uniform hypergraph $(V', E')$ with $V' = \{a_1, \hdots, a_4, b_1, \hdots, b_4, c_1, \hdots, c_4, d_1, \hdots, d_4\}$ and
\begin{align}
    E' := \{&e'_1 := (a_1,a_2,a_3,a_4),\ 
    e'_2 := (b_1, b_2, b_3, b_4),\ 
    e'_3 := (c_1, c_2, c_3, c_4),\ 
    e'_4 := (d_1, d_2, d_3, d_4),\nonumber\\
    &e'_5 := (a_1, b_2, c_3, d_4),\ 
    e'_6 := (b_1, a_2, d_3, c_4),\ 
    e'_7 := (c_1, d_2, a_3, b_4),\ 
    e'_8:= (d_1, c_2, b_3, a_4)\}.\label{eq:hypergraph-4}
\end{align}
Like $E$, we have that $E'$ is an Abelian cover because
\[
    \mathbf{1}_{e'_5} + \mathbf{1}_{e'_6} + \mathbf{1}_{e'_7} + \mathbf{1}_{e'_8} - \mathbf{1}_{e'_2} - \mathbf{1}_{e'_3} - \mathbf{1}_{e'_4} = \mathbf{1}_{e'_1}.
\]
However, we claim that $E'$ lacks a Catalan cover. Similar to \cite{brakensiek2025Richness}, we construct a suitable group $H$ and an extension map $\phi : V \to H$ such that the coset generated by $\{\phi(e'_2), \hdots, \phi(e'_8)\}$ fails to contain $\phi(e'_1)$ (with other cases following by symmetry). However, we found that the Pauli group considered in \cite{brakensiek2025Richness} does not appear to suffice for constructing such an extension. Rather, we needed to construct a considerably more sophisticated group $H$ for which such a coset non-containment occurs.

A key observation we make about the Pauli group considered by \cite{brakensiek2025Richness} is that it is (2-step) \emph{nilpotent}. That is, Given a group $G$, we can define its commutator $[G, G]$ to be the subgroup generated by $\{[g,h] := g^{-1}h^{-1}g h \mid g, h \in G\}$. Every Abelian group $G$ has the property that $[G, G]$ is the trivial group $1$, but 2-step nilpotent groups (such as the Pauli Group) have the weaker property that $[[G, G], G] = 1$. 

Implicit in the study of Abelian covers is the use of \emph{free Abelian groups}, where given a list of generators $\{g_1, \hdots, g_k\}$ we consider all group elements of the form $g_1^{\alpha_1} \cdots g_k^{\alpha_k}$ where $\alpha_1, \hdots, \alpha_k \in \mathbb Z$ and $[g_i, g_j] = 1$ for all $i, j \in [k]$. Likewise, one can construct a \emph{free 2-step nilpotent group} where all elements are of the form
\[
    \prod_{i=1}^k g_i^{\alpha_i} \prod_{1 \le i < j \le k} [g_i, g_j]^{\beta_{i,j}}.
\]
with $\alpha_1, \hdots, \alpha_k, \beta_{1,2}, \hdots, \beta_{k-1,k} \in \mathbb Z$~\cite{Tao09,CMZ17}. Our main technical contribution toward proving \cref{thm:4uniformIntro} is to show that if $H$ is the free 2-step nilpotent group with $k=4$ and we define the map $\phi$ to be 
\[
    \phi(a_i) = g_1, \phi(b_i) = g_2, \phi(c_i) = g_3, \phi(d_i) = g_4, \forall i \in [4],
\]
then the coset generated by $\{\phi(e'_2), \hdots, \phi(e'_8)\}$ fails to contain $\phi(e'_1)$. The proof proceeds by a lattice-based characterization of group elements appearing in the coset generated by $\{\phi(e'_2), \hdots, \phi(e'_8)\}$. Then, due to a certain parity condition, $\phi(e'_1)$ narrowly avoids being in this coset.

For applications to CSPs (such as CSP sparsification), it is useful to know whether the group $H$ can be made finite. We found that this is possible by adding to our free nilpotent group the relations that $g_i^4 = 1$ and $[g_i,g_j]^4=1$ for all $i,j \in [4]$. Such a group has order $4^{10}$. The analysis is nearly identical, except for keeping track of the modulo $4$ constraints.

\subsection{Outline}

In \cref{sec:abelian}, we prove \cref{thm:AbelianCoverIntro}, our main structural result on Abelian covers. In \cref{sec:Catalan}, we build up some important results on the structure of our ``Catalan stack group,'' leading to a straightforward proof of \cref{thm:CatalanArity3} in the graph setting. In \cref{sec:cover-3}, we prove \cref{thm:CatalanArity3} for 3-uniform hypergraphs, by way of various techniques in group theory and algebraic topology. In \cref{sec:limit}, we prove \cref{thm:4uniformIntro} that \cref{thm:CatalanArity3} cannot extend to arity $4$. In \cref{sec:csp}, we use the Catalan polymorphisms of \cite{brakensiek2025Richness} to show that \cref{thm:CatalanArity3} implies \cref{cor:arity3Maltsev}, our primary contribution to the theory of CSP non-redundancy.

\paragraph{AI Disclosure.} GPT 5.5 Pro was used to assist with the proof of \cref{thm:4uniformIntro}. More precisely, we had a proof that \cref{eq:hypergraph-4} has an extension into an infinite nilpotent group, which GPT helped to make finite by locating a suitable quotient. GPT was also used for help in generating some of the figures in the manuscript and finding a reference in the proof of \cref{thm:AbelianCoverIntro}. All text in the manuscript was written by the authors who take full responsibility for the content of the work.

We also used the computational tools Sagemath~\cite{sagemath,passagemath} and Z3~\cite{Z3} when finding the proof for \cref{thm:4uniformIntro}. However the proof of \cref{thm:4uniformIntro} presented here is self-contained.

\paragraph{Acknowledgments.} We thank anonymous SODA reviewers for many helpful comments which improved the presentation of this manuscript.

\section{On Abelian Covers}\label{sec:abelian}

We now formally define the notion of an Abelian cover. Our notation deviates a bit from the introduction. Let $V$ be a set of vertices and let $E \subseteq \binom{V}{r}$ be a set of $r$-uniform hyperedges. The hypergraph $H = (V, E)$ has a corresponding lattice $\langle H\rangle_{\Z} \le \Z^V$ defined as follows.
\[
\langle H\rangle_{\Z} := \left\{\sum_{h \in E} a_h\fe_h : a_h \in \mathbb Z\right\},
\]
where for $h = \{v_1, \hdots, v_r\} \in E$, we define $\fe_h = \fe_{v_1} + \cdots + \fe_{v_r} \in \mathbb Z^V$, where $\fe_{v}$ is the indicator vector of the $v$th coordinate. Likewise, for any commutative ring $R$ with a unit, we may define
\[
\langle H\rangle_{R} := \left\{\sum_{h \in E} a_h\fe_h : a_h \in R\right\},
\]
We say that $H$ has an \emph{Abelian cover} if there exists $h \in E$ such that\footnote{Here, $H \setminus \{h\}$ is shorthand for the hypergraph $(V, E \setminus \{h\})$. We also conflate a hypergraph which its edge set, so $h \in H$ is equivalent to $h \in E$.}
\[
\langle H\setminus \{h\}\rangle_{\Z} = \langle H\rangle_{\Z}.
\]
Otherwise, if no such $h$ exists, we say that $H$ is \emph{Abelian non-redundant}. For intuition, one can consider defining an analogous object $\langle H \rangle_{\Z_2}$, which is the span of the indicator vectors of hyperedges over $\Z_2$. In such an object the equivalent notion of being non-redundant would imply that no hyperedge in $H$ can be generated via sum of other hyperedges over $\Z_2$; equivalently, that there is no even cover of $H$. 

Our goal is to now prove \cref{thm:AbelianCoverIntro}. To do this, it suffices understand what is the maximal size of an Abelian non-redundant set.

\begin{lemma}\label{lemma:Z-NRD}
Let $H = (V,E)$ be a hypergraph on $n$ vertices which is Abelian non-redundant. Then, $|E| \le n\cdot (1 + \log_2(r)/2)$.
\end{lemma}

\begin{proof}
Let $n  := |V|$, and let us consider an $r$-partite hypergraph $H \subseteq \binom{V}{r}$ which is Abelian non-redundant of maximum size. 
In particular, these means that for every $h \in H$ that $h \notin \langle H\setminus \{h\}\rangle_{\Z}$ (i.e., that $h$ is not generated by the other hyperedges of $H$). 

First, we prove the theorem in the case when $\langle H\rangle_{\Q} \cap \Z^{n} = \Z^{n}$. 
In this case, we let $\hat{H} \subseteq H$ denote a minimal set of hyperedges such that $\langle H\rangle_{\Q}\cap \Z^{n} = \langle \hat{H}\rangle_{\Q}\cap \Z^{n} = \Z^{n}$ (i.e., $\hat{H}$ is a minimal generating set over $\Q$). Because $\Q$ is a field, the maximum size of such a generating set is bounded by $n$ (and in this case, in fact equal to $n$, as we generate $\Z^{n}$, which is a space of dimension $n$).

Next, we let $B_{\hat{H}} \in \zo^{n \times n}$ denote the matrix which is formed by those vectors $\fe_h: h \in \hat{H}$. Then, we know that (see, for instance, \cite{pinner2022lattices}), as a subgroup, the index of $\langle \hat{H}\rangle_{\Z}$ in $\Z^{n}$ is bounded as:
\begin{align}\label{eq:indexDet}
[\Z^{n}: \langle \hat{H}\rangle_{\Z}] = |\det(B_{\hat{H}})|,
\end{align}
and because each column of $\det(B_{\hat{H}})$ has at most $r$ $1$'s, we see that 
\[
|\det(B_{\hat{H}})| \leq \prod_{h \in \hat{H}} \Vert \fe_{\hat{h}}\Vert_2 \leq (\sqrt{r})^{n}.
\]

Now, let us fix an ordering of the remaining hyperedges in $H - \hat{H} = \{h_1, h_2, h_3, \dots h_{\ell}\}$, and let $\hat{H}^{(i)} = \hat{H} \cup \{h_1, \dots h_i\}$. By our assumption on $H$ being Abelian-NRD, it must be the case that for every $i \in [\ell]$ that 
\[
[\langle \hat{H}^{(i)}\rangle_{\Z}: \langle \hat{H}^{(i-1)}\rangle_{\Z}] \geq 2,
\]
as $h_i \notin \langle \hat{H}^{(i-1)}\rangle_{\Z}$, and thus $\hat{H}^{(i)}$ generates a strictly larger subgroup of $\Z^{n}$ than $\hat{H}^{(i-1)}$. In particular, this then means that 
\begin{align}\label{eq:totalDetDecrease}
[\langle H \rangle_{\Z}: \langle \hat{H}\rangle_{\Z}] = [\langle \hat{H}^{(\ell)}\rangle_{\Z}: \langle \hat{H}\rangle_{\Z}] \geq 2^{\ell}.
\end{align}

Now, using \cref{eq:indexDet} and \cref{eq:totalDetDecrease}, we see that 
\[
(\sqrt{r})^{n} \geq |\det(B_{\hat{H}})| = [\Z^{n}: \langle \hat{H} \rangle_{\Z}] = [\Z^{n}: \langle H \rangle_{\Z}] \cdot [\langle H \rangle_{\Z}: \langle \hat{H} \rangle_{\Z}] \geq 2^{\ell}.
\]
Thus, $\ell \leq \log(\sqrt{r}^{n}) = \frac{n \log(r)}{2}$, and thus the number of hyperedges in $H$ is at most $\ell + n \leq n \cdot (1 + \log(r)/2)$.

Now, we can generalize the above argument to the case when $\langle H\rangle_{\Q} \cap \Z^{n} \subsetneq \Z^{n}$. Again, we let $\hat{H} \subseteq H$ denote a minimal set of hyperedges such that $\langle H\rangle_{\Q}\cap \Z^{n} = \langle \hat{H}\rangle_{\Q}\cap \Z^{n}$. We set $d = \dim(\langle H\rangle_{\Q})$. 

Importantly, we now have the following claim:

\begin{fact}\label{clm:boundIndex}[See, for instance, Lemma 3.3 in \cite{goldberger2022practical}]
Let $a_1, \dots a_d \in \Z^{n}$ be linearly independent over $\Q$. Then, 
\[
[\langle (a_1, \dots a_d)\rangle_{\Q} \cap  \Z^{n}: \langle (a_1, \dots a_d)\rangle_{\Z}] \leq \prod_{j = 1}^d \Vert a_j \Vert_2.
\]
\end{fact}

\begin{remark}
    Note that \cite{goldberger2022practical} shows that $[\langle (a_1, \dots a_d)\rangle_{\Q} \cap  \Z^{n}: \langle (a_1, \dots a_d)\rangle_{\Z}]$ is bounded by the GCD of all $d \times d$ sub-determinants of the matrix $A \in \Z^{n \times d}$. In particular, this bounds the index by the determinant of a single such $d \times d$ submatrix, from which standard determinant bounds imply the stated fact.
\end{remark}

Now, we can apply \cref{clm:boundIndex} to our sub-hypergraph $\hat{H}$ to see that 
\[
[\langle H\rangle_{\Q} \cap \Z^{n}: \langle \hat{H}\rangle_{\Z} ] =  [\langle \hat{H}\rangle_{\Q} \cap \Z^{n}: \langle \hat{H}\rangle_{\Z} ]\leq (\sqrt{r})^d \leq (\sqrt{r})^{n}.
\]
From here, we can now repeat the same argument as above to deduce that there can be at most $\frac{n \log(r)}{2}$ additional hyperedges in the hypergraph before the hypergraph is Abelian-redundant. Thus, the number of hyperedges in an Abelian-NRD hypergraph is again bounded by $n + n \log(r) / 2$, yielding the desired lemma.
\end{proof}

\section{A Group-theoretic Perspective on Catalan Covers}\label{sec:Catalan}

Adapting the intuition from the introduction, we give now give a rigorous group-theoretic definition of Catalan covers, inspired by the observations of \cite{brakensiek2025Richness}. Recall from the introduction that for an $r$-uniform hypergraph, we maintain $r$ ``stacks'' of vertices to calculate Catalan covers. To keep track of which vertices are allowed in each stack, we assume that our $r$-uniform hypergraph is also $r$-partite. That is, we assume our variable set $V$ has a partition into $r$ sets $V_1, \hdots, V_r$. Given this vertex partition, we define the following finitely presented group, which we call an $r$-partite \emph{Catalan group}.\footnote{\cite{brakensiek2025Richness} considered a similar group (without the commutativity conditions) to prove that any CSP with a Mal'tsev extension has an infinite group extension. We discuss this and other CSP-specific details in \cref{sec:csp}.}
\[
    \G(V_1, \hdots, V_r) := \langle \{\fg_v : v \in V\} \mid \{\fg_v^2 = 1 : v \in V\} \cup \{\fg_v \fg_w = \fg_w \fg_v : v \in V_i, w \in V_j, i \neq j\}\rangle.
\]
For succinctness, we typically refer to the group as $\G$.\footnote{We typically use the \textsf{mathfrak} font (e.g., $\fg, \fw$) to denote group elements.}

Given $E \subseteq V_1 \times \cdots \times V_r$ and $r$-partite, $r$-uniform hypergraph $H \subseteq (V, E)$, we let $\langle H\rangle_{\G}$ be the subgroup generated by the words $\{\fg_h : h \in E\}$, where $\fg_h := \prod_{i=1}^r \fg_{v_i}$ if $h = (v_1, \hdots, v_r)$. We say that $H$ has a \emph{Catalan cover} if $\langle H \setminus \{h\}\rangle_{\G} = \langle H\rangle_{\G}$ for some $h \in E$. Otherwise, we say that $H$ is \emph{Catalan non-redundant}.

Before we can study the existence of Catalan covers, we first establish some basic structural results on the Catalan group $\G$. It is clear that any element of $\G$ can be written as \emph{words} of the form $\fg_{v_1} \cdots \fg_{v_k}$ where $v_1, \hdots, v_k \in V$ are vertices (possibly with repetition), but how can we easily tell if two words are equal too each other. In particular, could $\G$ be ``degenerate'' in the sense that $\G$ has only a single element? In this section, we build up a theory of $\G$ which is crucial not only for showing that $\G$ is well-defined, but also for eventually proving \cref{thm:CatalanArity3}.

\subsection{Reduced Words}\label{subsec:reduced-words}

The goal of this section is to define a notion of a \emph{reduced word} of $\G$ which lets tell when two words of $\G$ are equal or not. Unpacking the definition of $\G$, let
\[
\fF(V) := \langle \{\fg_v : v \in V\}\rangle
\]
be the free group with $V$ generators. Define
\[
  \fR(V_1, \hdots, V_r) := \langle \{\fg_v^2 : v \in V\} \cup \{[\fg_v, \fg_w] : v \in V_i, w \in V_j, i \neq j\} \rangle,
\]
where $[\fg_v, \fg_w] := \fg_v^{-1}\fg_w^{-1}\fg_v\fg_w$ is the standard group commutator. Further define $\ncl_{\fF}(\fR)$ be the smallest normal subgroup of $\fF$ containing $\fR$. The Catalan group $\G$ is then  precisely the quotient group $\G = \fF / \ncl_{\fF}(\fR)$~\cite{LyndonS01,HoltEO05}. In particular, if we seek to construct a homomorphism $\psi : \G \to G$, where $G$ is an arbitrary group, by the first isomorphism theorem, it suffices to construct a map $\widetilde{\psi} : \fF \to G$ such that $\fR \in \ker(\widetilde{\psi})$.

Despite $\G$ formally being a quotient group, we can reason about $\G$ much more concretely by looking at $\G$'s \emph{reduced words}. First, recall from the theory of free groups, a word $\fw \in \cF(V)$ is a sequence of symbols $\fw_1, \hdots, \fw_\ell \in \{\fg_v : v \in V\}  \cup \{\fg_v^{-1} : v \in V\}$. Furthermore, this word is \emph{reduced} if no two adjacent symbols of the word are inverses of each other. Each element $\fw \in \cF(V)$ is equal to a unique reduced word~\cite[Proposition 2.50]{HoltEO05}. We prove an analogous result for our Catalan group $\G$.

\begin{definition}\label{def:reduced-word-G}
A \emph{word} $\fw$ of $\G(V_1, \hdots, V_r)$ is a sequence $(\fw_1, \hdots, \fw_\ell) \in \{\fg_v : v \in V\}^{\ell}$ for some $\ell \ge 0$. We identify each word $\fw$ with its \emph{evaluation} $\fw_1\fw_2 \cdots \fw_\ell \in \G$. We let $1 = ()$ denote the trivial word corresponding to the group identity.

We say that $\fw$ is \emph{reduced} (with respect to $\G$) if (1) no two adjacent symbols are equal to each other and (2) for any $1 \le i < j \le \ell$ if $\fw_i \in V_a$ and $\fw_j \in V_b$, then $a \le b$. We say that $\fw'$ is a \emph{reduced form} of $\fw$ if $\fw' = \fw$ (in $\G$) and $\fw'$ is a reduced word.

Given a reduced word $\fw = \fw_1\cdots \fw_\ell$, its \emph{$r$-partite representation} is a partition $\fw = \fw^{(1)} \cdots \fw^{(r)}$ into subwords, where for each $i \in [r]$, $\fw^{(i)}$ has the symbols of $\fw$ inside $\{\fg_v : v \in V_i\}$. 
\end{definition}

\begin{remark}
By the relations defining $\G$, observe that each of $\fw^{(1)}, \hdots, \fw^{(r)}$ is itself reduced and that their evaluations commute with each other.
\end{remark}

\begin{example}
Let $V_1 = \{1,2\}$, $V_2 = \{3,4\}$, and $V_3 = \{5,6\}$. Then $1$, $\fg_1\fg_2\fg_3\fg_4\fg_5\fg_6$, and $\fg_1\fg_2\fg_1\fg_2\fg_4\fg_5\fg_6\fg_5$ are reduced words of $\G$. Their $r$-partite representations are equal to $(1, 1, 1)$, $(\fg_1\fg_2, \fg_3\fg_4, \fg_5\fg_6)$, and $(\fg_1\fg_2\fg_1\fg_2, \fg_4, \fg_5\fg_6\fg_5)$, respectively. However, $\fg_1^2\fg_2\fg_3$ and $\fg_3\fg_2\fg_1\fg_4$ are not reduced words of $\G$.
\end{example}

We first show that every $g \in \G$ equals at least one reduced word.

\begin{proposition}\label{prop:G-at-least-one-word}
For every $g \in \G$, there is at least one reduced word $w$ such that $g = w$.
\end{proposition}

\begin{proof}
Any $\fg \in \G$ is a coset of words $\fw \in \cF(V)$. Pick one such $\fw = \fw_1 \cdots \fw_\ell$ arbitrarily (i.e., $\fg = \fw + \ncl_{\fF}(\fR)$). Let $\hat{\fw} = \hat{\fw}_1 \cdots \hat{\fw}_\ell$ be the corresponding word:
\[
  \hat{\fw}_i = \begin{cases}
    \fg_v & \fw_i = \fg_v\\
    \fg_v & \fw_i = \fg^{-1}_v.
  \end{cases}
\]
Since $\G$ satisfies the relations $\fg_v^2 = 1$ for all $v \in V$, we have that $\hat{\fw} + \ncl_{\fF}(\fR) = w + \ncl_{\fF}(\fR) = g$. Note that $\hat{\fw}$ (but not necessarily $\fw$) is a word equal to $g$ in the sense of \cref{def:reduced-word-G}.

Now, for every $i \in [\ell]$, assume that $\hat{\fw}_i = \fg_{v_i}$ and that $v_i \in V_{a_i}$. Consider the sequence of pairs $((a_1, \hat{\fw}_1), \hdots, (a_\ell, \hat{\fw}_\ell))$ and let $\sigma \in S_\ell$ correspond to the unique stable sort $((a_{\sigma(1)}, \hat{\fw}_{\sigma(1)}), \hdots,$ $(a_{\sigma(\ell)}, \hat{\fw}_{\sigma(\ell)}))$. That is, for any $1 \le i < j \le n$ we have that $a_{\sigma(i)} \le a_{\sigma(j)}$ and furthermore if $a_{\sigma(i)} = a_{\sigma(j)}$ then $\sigma(i) < \sigma(j)$. It is well-known (e.g., \cite[Exercise 2-2]{CLRS09}) that a stable sort can be done using only adjacent swaps. By stability, each swap is only going to replace $(\hat{\fw}_i, \hat{\fw}_j)$ with $(\hat{\fw}_j, \hat{\fw}_i)$ if $a_i > a_j$. Since $a_i \neq a_j$, we have that $[\fg_{v_i}, \fg_{v_j}] = 1$ is a relation of $\G$, so this swap preserves the value of the word in $\G$. Call the stably sorted word $\tilde{\fw}$ where $\tilde{\fw}_i = \hat{\fw}_{\sigma(i)}$ for all $i \in [\ell]$.

This word $\tilde{\fw}$ satisfies property (2) of being reduced, but it may not satisfy property (1). To finish, if two adjacent symbols of $\tilde{\fw}$ are equal, we may delete both without changing the value the word as $\fg_v^2 = 1$ for all $v \in V$. Since $\tilde{\fw}$ has finite length, this operation is only done finitely many times. When no more cancellations are possible, our word satisfies both (1) and (2), where (2) holds because deleting symbols cannot remove the sorted order.
\end{proof}

\begin{remark}\label{rem:naive-reduction}
Given a word $\fw$ satisfying (2) but not necessarily (1) of \cref{def:reduced-word-G}, we can apply the procedure discussed in the final paragraph of the proof of \cref{prop:G-at-least-one-word} to get a reduced word $\fw'$. We call this $\fw'$ the \emph{naive} reduction of $\fw'$. As we soon prove in \cref{thm:G-reduced-not-equal}, this is the \emph{only} reduced word which evaluates to the same group element was $\fw$.
\end{remark}

We now show this reduced representation is unique. We begin by showing that non-empty distinct words cannot equal the identity.

\begin{proposition}\label{prop:G-reduced-not-1}
Every non-empty reduced word of $\G$ does not evaluate to the identity.
\end{proposition}
\begin{proof}

Let $\fw$ be a non-empty reduced word. To prove that $\fw$ cannot evaluate to the identity, we construct a (finite) group $G$ and a homomorphism $\psi : \fF(V) \to G$ such that $\fR(V_1, \hdots, V_r) \subseteq \ker(\psi)$ but $\psi(\fw) \neq 1$ (where $\fw$ is interpreted as a word of $\fF$).

Our construction is inspired by \cite[Exercise 2.7]{CSC23}. Let $\fw$ be a non-empty reduced word and let its $r$-partite representation be
\[
\fw = \fw^{(1)} \fw^{(2)} \cdots \fw^{(r)} = \fw_1^{(1)} \cdots \fw_{\ell_1}^{(1)} \fw_1^{(2)} \cdots \fw_{\ell_2}^{(2)} \cdots \fw_1^{(r)} \cdots \fw_{\ell_r}^{(r)}.
\]
For each $i \in [r]$ and $j \in [\ell_i]$, let $v^{(i)}_j \in V_i$ be such that $\fw_j^{(i)} = \fg_{v^{(i)}_j}$. Given $v \in V_i$, let $I_v$ be the set of indices $j \in [\ell_i]$  such that $v = v_j^{(i)}$.

We can now construct our group $G$ and the homomorphism $\psi : \fF(V) \to G$. Let $L$ be the following label set
\[
  L = \bigcup_{i=1}^r \{i\} \times [\ell_i + 1].
\]
We then let $G = S_L$, the symmetric group on the set $\ell$. For every $i \in [r]$ and $v \in V_i$, let $\sigma_v \in S_L$ be the following permutation
\[
  \sigma_v(j, a) = \begin{cases}
    (j, a+1) & j = i\text{ and }a \in I_v\\
    (j, a-1) & j = i\text{ and }a-1 \in I_v\\
    (j, a) & \text{otherwise}.
    \end{cases}
\] permutation which swaps $V_i$. Observe that since $w$ is reduced, no two adjacent symbols of $w$ are equal. Thus, the difference between consecutive elements of $I_v$ is at least $2$. Thus, $\sigma_v$ defined a product of transpositions, so $\sigma_v^2 = 1$. Furthermore, if $v \in V_i$ and $v' \in V_j$ with $i \neq j$, then the elements of $L$ permuted by $\sigma_v$ and $\sigma_{v'}$ are disjoint. Thus, $[\sigma_v, \sigma_{v'}] = 1$. In other words, we have that the map $\psi : \fF(V) \to S_L$ for which $\psi(\fg_v) = \sigma_v$ is a homomorphism for which $\fR \in \ker(\psi)$. Therefore, $\psi$ is a well-defined homomorphism from $\G$ to $S_L$.

To finish, it remains to show that $\psi(\fw) \neq 1$. Since $w \neq 1$, there is $i \in [r]$ for which $\fw^{(i)} \neq 1$. Now observe that
\begin{align*}
  \psi(\fw)(i, \ell_i+1) &= \psi(\fw^{(i)})(i, \ell_i+1)\\
                       &= (\sigma_{v_1^{(i)}} \circ \cdots \circ \sigma_{v_{\ell_i}^{(i)}})(i, \ell_I+1)\\
                       &=(\sigma_{v_1^{(i)}} \circ \cdots \circ \sigma_{v_{\ell_i-1}^{(i)}})(i, \ell_I)\\
                       &\cdots = \sigma_{v_1^{(i)}}(i, 2) = (i, 1).
\end{align*}
Therefore, $\psi(\fw) \neq 1$, as desired.
\end{proof}

Before proving our main result of this subsection that no two reduced words are equal (\cref{thm:G-reduced-not-equal}), we use \cref{prop:G-reduced-not-1} to show that $\G$ is \emph{residually finite} (see \cite[Definition 2.1.1]{CSC23}). That is, for every $\fg \in \G \setminus \{1\}$, there exists a finite group $G$ and a homomorphism $\phi : \G \to G$ such that $\phi(\fg) \neq 1$.

\begin{corollary}\label{prop:G-profinite}
$\G$ is residually finite.
\end{corollary}
\begin{proof}
For every $\fg \in \G \setminus \{1\}$, by \cref{prop:G-at-least-one-word} there exists a reduced word $\fw$ equal to $\fg$. Since $\fw \neq 1$, by the proof of \cref{prop:G-reduced-not-1}, we have there is a homomorphism $\phi : \G \to G$ (where $G$ is finite) such that $\phi(\fw) = \phi(\fg) \neq 1$.
\end{proof}

\begin{theorem}\label{thm:G-reduced-not-equal}
No two distinct reduced words of $\G$ evaluate to equal elements $\G$.
\end{theorem}
\begin{proof}
We give a proof by construction. Let $\fw = \fw_1\cdots \fw_\ell$ and $\fw' = \fw'_1 \cdots \fw'_{\ell'}$ be distinct reduced words of $\G$ such that their evaluations are equal. Further assume that among counterexamples that $\ell+\ell' \ge 1$ is minimal. Our goal is to show there is a non-empty reduced word equal to $1$ in $\G$, which contradicts \cref{prop:G-reduced-not-1}.

Following \cref{def:reduced-word-G}, since $\fw$ and $\fw'$ are reduced, write $\fw = \fw^{(1)} \cdots \fw^{(r)}$ and $\fw' = \fw'^{(1)} \cdots \fw'^{(r)}$ for their $r$-partite representations. Note that $\fw^{(1)}, \hdots, \fw^{(r)}$ are reduced words as they are substrings of a reduced word.

As one piece of additional notation, for all $i \in [r]$, we let $\rev(\fw'^{(i)})$ denote $\fw'^{(i)}$ written in reverse. Since $\fw'^{(i)}$ is reduced and all symbols of $\fw'^{(i)}$ lie in $\{\fg_{v} : v \in V_i\}$, we have that $\rev(\fw'^{(i)})$ is also a reduced word. Also observe that the evaluation of $\rev(\fw'^{(i)})$ is precisely the inverse of $\fw'^{(i)}$. Now consider the word
\[
  \hat{\fw} = \rev(\fw'^{(1)})\fw^{(1)}\rev(\fw'^{(2)})\fw^{(2)}\cdots \rev(\fw'^{(r)})\fw^{(r)}.
\]
By the commutativity rules of $\G$, we have that $\hat{\fw}$ equals $(\fw')^{-1}\fw = 1$ in $\G$. However, we claim that $\hat{\fw}$ is reduced. Since condition (2) of \cref{def:reduced-word-G} is enforced by $\G$, it suffices to check that condition (1) holds. Since each subword (i.e., $\fw^{(i)}$ or $\fw'^{(i)}$) is reduced, the only potential violation of (1) is between subwords. More precisely, the first symbol of $\fw^{(i)}$ and $\fw'^{(i)}$ equal to some $\fg_{v}$ (with $v \in V_i$) for some $i \in [r]$. Let $\tilde{\fw}$ and $\tilde{\fw}'$ be the words $\fw$ and $\fw'$, respectively, except this first occurrence of $\fg_v$ is deleted from both. By the commutativity rules of $\G$, we can see that in $\G$ we have that
\[
  \tilde{\fw} = \fg_v\fw = \fg_v\fw' = \tilde{\fw}'.
\]
Furthermore, since $\fw$ and $\fw'$ are distinct reduced words, we have that $\tilde{\fw}$ and $\tilde{\fw}'$ are also distinct reduced words. This violates that the pair $(\fw, \fw')$ is a minimal counterexample. In summary, we have found a non-empty reduced word $\hat{\fw}$ which is equal to $1$ in $\G$, this contradicts \cref{prop:G-reduced-not-1}.\end{proof}

\subsection{Parity Homomorphism}\label{subsec:parity}

Recall from the introduction that a ``stack'' of hyperedges $h_1, \hdots, h_k \in E$ can be modeled as $\fg_{h_1} \cdots \fg_{h_k}$. In general, such products cannot produce every element of $\G$. To keep track of which elements can be generated, we look at a family of homomorphisms related to $\G$ which we call \emph{parity homomorphisms}. For any $i \in [r]$, consider the unique homomorphism $\pi_i : \cF(V) \to \Z_2$ such that
\[
  \pi_i(\fg_v) := \begin{cases}
    1 & v \in V_i\\
    0 & \text{otherwise}.
  \end{cases}
\]
By inspection, $\fR(V_1, \hdots, V_r) \subseteq \ker(\pi_i)$ (e.g., $\pi_i(\fg_v^2) = 2\pi_i(\fg_v) = 0$). Therefore, with a slight abuse of notation, $\pi_i : \G \to \mathbb Z_2$ is a homomorphism. Given a subset $S \subseteq [r]$, we define $\pi_S : \G \to \Z_2^S$ via $\pi_S(x) = (\pi_i(x) : i \in S)$. We show that $\pi_{[r]}$ is an epimorphism (i.e., a surjective homomorphism).

\begin{proposition}\label{prop:pi-surjective}
  Assuming $V_1, \hdots, V_r$ are nonempty, the map $\pi_{[r]} : \G \to \Z_2^r$ is surjective.
\end{proposition}
\begin{proof}
Pick $v_1 \in V_1, \hdots, v_r \in V_r$ arbitrarily. For any $b \in \Z_2^r$, consider the word $\fw_b := \prod_{i=1}^{r} \fg_{v_i}^{b_i}$. It is straightforward to verify that $\pi_{[r]}(\fw_b) = b$. Thus, the coset $\fw_b + \ncl_{\fF}(\fR) \in \G$ also maps to $b$, so $\pi_{[r]}$ is surjective.
\end{proof}

As an application of this proposition, let $H = (V,E)$ with $E = V_1 \times \cdots \times V_r$ be the complete $r$-partite hypergraph. Consider the group
\[
  \fV := \langle H\rangle_{\G}.
\]
By definition $\fV$ is a subgroup of $\G$. We now show that $\fV$ is a normal subgroup of finite index.

\begin{proposition}
We have that
\begin{align}
  \fV = \{g \in \G : (\pi_1+\pi_2)(g) = 0, (\pi_1+\pi_3)(g) = 0, \hdots, (\pi_1+\pi_r)(g) = 0\},\label{eq:V-eq}
\end{align}
is a normal subgroup of $\G$. Furthermore, $\G / \fV \cong \Z_2^{r-1}$.
\end{proposition}

\begin{proof}
Consider each generator $\fg_h$ of $\fV$, where $h = (v_1, \hdots, v_r) \in V_1 \times \cdots \times V_r$. Since $\fg_h = \fg_{v_1}\fg_{v_2} \cdots \fg_{v_r}$, we have that $\pi_{[r]}(\fg_h) = (1, 1, \hdots, 1) \in \Z_2^r$. Thus, $(\pi_1+\pi_2)(\fg_h) = 0, (\pi_1+\pi_3)(\fg_h) = 0, \hdots, (\pi_1+\pi_r)(\fg_h) = 0$. Thus, the LHS of \cref{eq:V-eq} is contained in the RHS.

To show the converse, we use the framework of reduced words in $\G$ established in \cref{subsec:reduced-words}. We give a proof by contradiction consider an arbitrary reduced word $\fw \in \G$ with $r$-partite representation $\fw = \fw^{(1)} \cdots \fw^{(r)}$ such that $(\pi_1+\pi_2)(\fw) = 0, (\pi_1+\pi_3)(\fw) = 0, \hdots, (\pi_1+\pi_r)(\fw) = 0$ but $\fw \not\in \fV$. We assume that $\fw$ has minimal length among counterexamples.

We break up our analysis into cases. If $\fw^{(i)} \neq 1$ for all $i \in [r]$, then consider $h \in V_1 \times \cdots \times V_r$, where $\fw^{(i)}_1 = \fg_{h_i}$ for all $i \in [r]$. Then, since $\fw \not\in \fV$, we have that $\fg_h\fw \not\in \fV$. However, the unique form of $\fg_h\fw$ has $r$ fewer symbols, a contradiction.

Otherwise, assume that $\fw^{(i)} = 1$ for some $i$. In particular, $\pi_i(\fw) = 0$. Thus, since  $(\pi_1+\pi_2)(\fw) = 0, (\pi_1+\pi_3)(\fw) = 0, \hdots, (\pi_1+\pi_r)(\fw) = 0$, we may deduce that $\pi_i(\fw) = 0$ for all $i \in [r]$. The case $w=1$ is an obvious contradiction, so assume that $\fw^{(i)} \neq 1$ for some $i \in [r]$. Since $\pi_i(\fw) = 0$, having $\fw^{(i)} \neq 1$ implies $\fw^{(i)}$ has at least two symbols in its reduced form.

Pick $v_1 \in V_1, \hdots, v_r \in V_r$ arbitrarily. Define words $\fu^{(1)} = (\fu^{(1)}_1, \hdots, \fu^{(1)}_r), \fu^{(2)} = (\fu^{(2)}_1, \hdots, \fu^{(2)}_r) \in \fV$ as follows:
\[
\fu^{(a)}_i = \begin{cases}
  \fw^{(i)}_a & \fw^{(i)} \neq 1\\
  \fg_{v_i} & \text{otherwise}.
  \end{cases}
\]
Now $\fw'$ be the unique reduced word which evaluates to $\fu^{(2)}\fu^{(1)}\fw \not\in \fV$. Let $\fw'^{(1)} \cdots \fw'^{(r)}$ be the $r$-partite representation of $\fw'$. If $\fw^{(i)} = 1$, it is straightforward to check that $\fw'^{(i)} = 1$ as well. Otherwise, if $\fw^{(i)} \neq 1$, then $\fw'^{(i)}$ has two fewer symbols that $\fw^{(i)}$. Since $\fw \neq 1$, this implies that $\fw'$ has fewer symbols than $\fw$, a contradiction. Thus, \cref{eq:V-eq} holds.

So far, we have shown that $\fV$ is precisely the kernel of $\G$ with respect to the map $(\pi_1+\pi_2, \hdots, \pi_1+\pi_r)$. Since $(\pi_1, \hdots, \pi_r)$ is a surjection of $\G$ by \cref{prop:pi-surjective}, we can see that $(\pi_1+\pi_2, \hdots, \pi_1+\pi_r)$ is a surjection of $\G$. Thus, by the first isomorphism theorem, we have that $\G /\fV \cong \im(\pi_1+\pi_2, \hdots, \pi_1+\pi_r) = \Z_2^{r-1}$, as desired.
\end{proof}

\subsection{Proof of \texorpdfstring{\cref{thm:CatalanArity3}}{Theorem 1.2} for \texorpdfstring{$r=2$}{r=2}}

To get help the reader better understand the Catalan group, we show how that $r=2$ case of \cref{thm:CatalanArity3} quickly follows. It is straightforward to show the following lemma suffices.

\begin{lemma}\label{lem:cycle-redundant}
Let $E \subseteq V_1 \times V_2$ be a bipartite cycle. Then, $H = (V, E)$ has a Catalan cover.
\end{lemma}
\begin{proof}
Denote the vertices traversed by the cycle $E$ as $\ell_1, r_1, \ell_2, r_2, \ell_3, \dots r_{|H|/2}$. We claim that $\langle H \setminus (\ell_1, r_{|H|/2}) \rangle_{\G} = \langle H \rangle_{\G}$. To see why, we consider the sequence of edges 
\[ (\ell_1, r_1), (\ell_2, r_1), (\ell_2, r_2), \dots (\ell_{|H|/2},  r_{|H|/2}) \ . \] 
The word generated by these edges is exactly 
\[
\fg_{\ell_1}\fg_{r_1} \cdot \prod_{i = 2}^{|H|/2} \fg_{\ell_i}\fg_{r_{i-1}} \fg_{\ell_i}\fg_{r_{i}} = \fg_{\ell_1}\fg_{r_1} \cdot \prod_{i = 2}^{|H|/2} \fg_{r_{i-1}}\fg_{r_{i}} = \fg_{\ell_1} \cdot \left ( \prod_{i = 1}^{|H|/2 -1 } \fg_{r_{i-1}}^2  \right ) \fg_{r_{|H|/2}} = \fg_{\ell_1}\fg_{r_{|H|/2}},
\]
and thus $(\ell_1, r_{|H|/2}) \in \langle H \setminus (\ell_1, r_{|H|/2}) \rangle_{\G}$. Therefore,  $\langle H \setminus (\ell_1, r_{|H|/2}) \rangle_{\G} = \langle H \rangle_{\G}$, proving that $H$ has a Catalan cover.
\end{proof}

\subsection{Ruling out Catalan Covers using Group Extensions}

\cref{lem:cycle-redundant} is a great example of proving that a hypergraph has a Catalan cover, you build a suitable sequence of hyperedges and do a (usually simple) computation in the Catalan group $\G$. However, how do we prove that a hypergraph \emph{lacks} a Catalan cover; that is, it is Catalan non-redundant? In short, we can do this ruling out an extension into a simpler group. We formalize this concept as follows, which we make use of in \cref{thm:4uniformIntro}.

Let $H = (V, E)$ be an $r$-partite $r$-uniform hypergraph. Let $\fN$ be an arbitrary group. Given an arbitrary map $\psi : V \to \fN$, we define for each $h = (v_1, \hdots, v_r) \in E$ a word $\psi(h) \in \fN^r$ as follows
\[
  \psi(h) = (\psi(v_1), \hdots, \psi(v_r)).
\]
We let $\langle H\rangle_{\psi, \fN}$ denote the subgroup of $\fN^r$ generated by $\{\psi(h) \mid h \in E\}$.

\begin{proposition}\label{prop:test}
If there exists $\fN$ and $\psi : V \to \fN$ such that $\langle H \setminus \{h\} \rangle_{\psi, \fN} \subsetneq \langle H\rangle_{\psi, \fN}$, then $\langle H \setminus \{h\} \rangle_{\G} \subsetneq \langle H\rangle_{\G}$.
\end{proposition}

\begin{proof}
Assume for sake of contradiction that $\langle H \setminus \{h\} \rangle_{\G} = \langle H\rangle_{\G}$. By \cref{prop:pi-surjective}, this implies there exists an odd integer $k \ge 1$ and a sequence $h_1, \hdots, h_k \in E$ of hyperedges (possibly with repetition) such that
\[
  \fg_{h_1} \cdots \fg_{h_k} = \fg_h.
\]
Assume that $h = (v_1, \hdots, v_r)$ and $h_i = (v_{i,1}, \hdots, v_{i,r})$ for all $i \in [k]$. Then by \cref{thm:G-reduced-not-equal}, we have that for all $j \in r$, the word
\[
 \prod_{i=1}^k \fg_{v_{i,j}}
\]
reduces to $\fg_{v_j}$. Combinatorially, this reduction procedure always cancels out equal symbols, one of which has even index and the other has odd index. Thus, we can apply the exact same reduction procedure to the sequence $(\psi(v_{1,j}), \psi(v_{2,j}), \hdots, \psi(v_{k,j})) = \psi(v_j)$. As a result, we can deduce that
\[
  \prod_{i=1}^k \psi(v_{i,j})^{(-1)^{i-1}} = \psi(v_j),
\]
where the reduction in the Catalan group is simulated in $\fN$ via the identity $\fn \cdot \fn^{-1} = \fn^{-1} \cdot \fn = 1$ for all $\fn \in \fN$. Therefore, we can more globally deduce that
\[
\prod_{i=1}^k \psi(h_i)^{(-1)^{i-1}} = \psi(h),
\]
so $\langle H \setminus \{h\} \rangle_{\psi, \fN} = \langle H\rangle_{\psi, \fN}$, a contradiction.
\end{proof}

\begin{remark}
Letting $\fN = \G$ and $\psi$ the map $\psi(v) = \fg_v$, we can use the theory in \cref{subsec:reduced-words} and \cref{subsec:parity} to prove that $\langle H \setminus \{h\}\rangle_{\G} = \langle H\rangle_{\G}$ if and only if $\langle H \setminus \{h\}\rangle_{\psi, \fN} = \langle H\rangle_{\psi, \fN}$. The proof follows form observing that $\langle H\rangle_{\psi,\fN} \cong \langle H\rangle_{\G}$ via the map $(\psi(e_1), \hdots, \psi(e_k)) = (\fg_{e_1}, \cdots, \fg_{e_k}) \mapsto \fg_{e_1} \cdots \fg_{e_k}$ for all $e \in E$. We point out this observation to emphasize that \cref{prop:test} can always be used to prove  a hypergraph $H$ lacks a Catalan cover.
\end{remark}

\section{Characterization of Catalan Covers When \texorpdfstring{$r=3$}{r=3}}\label{sec:cover-3}

With the $r=2$ case of \cref{thm:CatalanArity3} taken care of, we now turn toward resolving the $r=3$ case. Similar to previous sections we assume our tripartite vertex set is $V = V_1 \cup V_2 \cup V_3$. Consider $\mathfrak G := \mathfrak G(V_1, V_2, V_3)$. We seek to prove the following. Note that in this section we conflate the notation for a hypergraph with its hyperedge set. 

\begin{theorem}\label{thm:3-Catalan}
Consider $H \subseteq H' \subseteq V_1 \times V_2 \times V_3$ such that $\langle H\rangle_{\mathbb Z} = \langle H'\rangle_{\mathbb Z}$, then $\langle H\rangle_{\fG} = \langle H'\rangle_{\fG}$.
\end{theorem}

\subsection{The Label Group}

 We begin by defining the \emph{label group} which is motivated in \cref{subsec:topology-intro}. We let $\mathfrak L := \mathfrak G(V_3) \subseteq \mathfrak G$.  Define $\langle H\rangle_{\mathfrak L} := \langle H\rangle_{\mathfrak G} \cap \mathfrak L$. A crucial observation is that $\langle H\rangle_{\fL}$ is a normal subgroup of $\fL$.

\begin{proposition}\label{prop:normal-subgroup}
Assume that $H \subseteq V_1 \times V_2 \times V_3$ has at least one edge incident with each $v \in V_3$. Then $\langle H\rangle_{\mathfrak L} \lhd \mathfrak L$.
\end{proposition}

\begin{proof}
It suffices to prove for each $v \in V_3$ and $\fw \in \langle H\rangle_{\mathfrak L}$ we have that $\fg_v \fw \fg_v \in \langle H\rangle_{\mathfrak L}$. By our assumption on $H$, there exists an edge $(t, u, v) \in H$. Furthermore, since $\fw \in \fL$, we have that $\fw$ commutes with $\fg_t$ and $\fg_u$. Since $\fw \in \langle H\rangle_{\mathfrak G}$, we have that
\[
    \fg_t\fg_u\fg_v \fw \fg_t\fg_u\fg_v \in \langle H\rangle_{\mathcal G}.
\]
Furthermore, by commutativity, we have that
\[
    \fg_t\fg_u\fg_v \fw \fg_t\fg_u\fg_v = \fg_t^2 \fg_u^2 \fg_v \fw \fg_v = \fg_v \fw \fg_v \in \fL.
\]
Thus,  $\fg_v \fw \fg_v \in \langle H\rangle_{\mathfrak L}$, as desired.
\end{proof}

In particular, \cref{prop:normal-subgroup} implies that we can study the quotient $\fL / \langle H\rangle_{\fL}$. This shall be quite useful after we build a bit more machinery.

We now identify an important group of elements of $\langle H\rangle_{\fL}$ corresponding to \emph{$(1,2)$-cycles}. Given $H \subseteq V_1 \times V_2 \times V_3$, we let $E_{1,2}(H)$ be the \emph{multiset} $\{(v_1, v_2) \mid (v_1, v_2, v_3) \in H\}$. We call a cycle in $E_{1,2}(H)$ a \emph{$(1,2)$-cycle} of $H$. With a slight abuse of notation, we say that for an even integer $\ell$ we say that a sequence of edges $C := \{(a_1, b_1, c_1), (a_2, b_2, c_2), \hdots, (a_{\ell}, b_{\ell}, c_{\ell})\} \subseteq H$ is a $(1,2)$-cycle if $a_i = a_{i+1}$ for all odd $i \in [\ell]$ and $b_{i} = b_{(i+1)\mod \ell}$ for all even $i \in [\ell]$. We observe the following
\begin{proposition}\label{prop:12-cycle}
Let $C := \{(a_1, b_1, c_1), (a_2, b_2, c_2), \hdots, (a_{\ell}, b_{\ell}, c_{\ell})\} \subseteq H$ be a $(1,2)$-cycle. Then, $\fg_{\mathcal C} := \fg_{c_1} \fg_{c_2} \cdots \fg_{c_\ell} \in \langle H\rangle_{\fL}$.
\end{proposition}

\begin{proof}
Clearly $\fg_{\mathcal C} \in \fL$. Thus, it suffices to prove that $\fg_{\mathcal C} \in \langle H\rangle_{\fG}$. Let 
\[
  \fw := \prod_{i=1}^{\ell}\fg_{a_i}\fg_{b_i}\fg_{c_i} \in \langle H\rangle_{\fG}.
\]
We claim that $\fw = \fg_{\mathcal C}$. In particular, using commutativity of the generators of $\fG$ and the fact that $\fg_v^2 = 1$ for all $v \in V$, we have that
\begin{align*}
\fw &= (\fg_{a_1}\fg_{a_2}\fg_{a_3} \cdots \fg_{a_\ell})(\fg_{b_1}\fg_{b_2}\fg_{b_3} \cdots \fg_{b_\ell})\fg_{\mathcal C}\\
&= (\fg_{a_1}^2 \fg_{a_3}^2 \cdots \fg_{a_{\ell-1}}^2)(\fg_{b_1} \fg_{b_2}^2 \cdots \fg_{b_{\ell-1}}^2 \fg_{b_{\ell}})\fg_{\mathcal C}\\
&= \fg_{b_1} \fg_{b_\ell} \fg_{\mathcal C}\\
&= \fg_{\mathcal C}.\qedhere
\end{align*}
\end{proof}

We can use \cref{prop:12-cycle} to answer an important question. Assume that $H \subsetneq H' \subseteq V_1 \times V_2 \times V_3$ and $\langle H\rangle_{\fG} \subsetneq \langle H'\rangle_{\fG}$.

\begin{proposition}\label{prop:G-to-L}
Consider $H \subsetneq H' \subseteq V_1 \times V_2 \times V_3$ such that $E_{1,2}(H)$ and $E_{1,2}(H')$ have the same connected components. Then, if $\langle H\rangle_{\fG} \subsetneq \langle H'\rangle_{\fG}$, we have that $\langle H\rangle_{\fL} \subsetneq \langle H'\rangle_{\fL}$.
\end{proposition}

\begin{proof}
Inclusion is trivial, so it suffices to identify some $w \in \langle H'\rangle_{\fL} \setminus \langle H\rangle_{\fL}$. Since $\langle H\rangle_{\fG} \subsetneq \langle H'\rangle_{\fG}$, there exists some $e \in H' \setminus H$ for which $\fg_e \not\in \langle H\rangle_{\fG}$.  Since $E_{1,2}(H)$ and $E_{1,2}(H')$ have the same connected components there exists a cycle in $E_{1,2}(H \cup \{e\})$ that uses $e$ exactly once. In other words, there is an even integer $\ell$ and edges
\[
(a_1, b_1, c_1), (a_2, b_2, c_2), \hdots, (a_{\ell-1}, b_{\ell-1}, c_{\ell-1}) \in H,\text{ and } (a_\ell, b_\ell, c_\ell) = e
\]
such that $a_i = a_{i+1}$ for all odd $i \in [\ell]$ and $b_{i} = b_{i+1 \mod \ell}$ for all even $i \in [\ell]$. Now consider the word $w \in \langle H'\rangle_{\fG}$ defined by
\[
    w := \prod_{i=1}^{\ell}\fg_{a_i}\fg_{b_i}\fg_{c_i}.
\]
By \cref{prop:12-cycle}, we have that
\[
    w = \fg_{c_1}\fg_{c_2} \cdots \fg_{c_\ell} \in \fL.
\]
Thus, $w \in \langle H'\rangle_{\fL}$. We claim that $w \not\in \langle H\rangle_{\fL}$. Assume for sake of contradiction that $w \in \langle H\rangle_{\fL}$. Then, observe that
\begin{align*}
    \langle H\rangle_{\fG} &\ni \left[\prod_{i=\ell-1}^{1}\fg_{a_i}\fg_{b_i}\fg_{c_i}\right] \cdot \prod_{i=1}^{\ell} \fg_{c_i} \\
    &= \fg_{a_{\ell-1}} \fg_{b_1} \fg_{c_\ell}\\
    &= \fg_{a_\ell} \fg_{b_\ell} \fg_{c_\ell}\\
    &= \fg_e,
\end{align*}
but we already assumed that $\fg_e \not\in \langle H\rangle_{\fG}$, a contradiction. Therefore, $\langle H\rangle_{\fL} \subsetneq \langle H'\rangle_{\fL}$.
\end{proof}

We note that the hypothesis of \cref{prop:G-to-L} is satisfied when $\langle H\rangle_{\mathbb Z} = \langle H'\rangle_{\mathbb Z}$.

\begin{proposition}\label{prop:same}
Assume that $H \subseteq H' \subseteq V_1 \times V_2 \times V_3$ and $\langle H\rangle_{\mathbb Z} = \langle H'\rangle_{\mathbb Z}$, then $E_{1,2}(H)$ and $E_{1,2}(H')$ have the same connected components.
\end{proposition}
\begin{proof}
To see why, assume that $e = (v_1, v_2, v_3)$ and that $(v_1, v_2)$ bridges two components of $E_{1,2}(H)$. Further, let $(a_1, b_1, c_1), \hdots, (a_N, b_N, c_N) \in H$ and weights $\alpha_1, \hdots, \alpha_N \in \mathbb Z$ such that
\[
 \sum_{i=1}^N \alpha_i (\fe_{a_i} + \fe_{b_i} + \fe_{c_i}) = \fe_{v_1} + \fe_{v_2} + \fe_{v_3}.
\]
Since the above identity is a linear combination of basis vectors, we can in fact deduce that
\begin{align*}
 \sum_{i=1}^N \alpha_i \fe_{a_i} &= \fe_{v_1},\text{ and}\\
 \sum_{i=1}^N \alpha_i \fe_{b_i} &= \fe_{v_2}
\end{align*}
In particular, this means for any function $f : V_1 \cup V_2 \to \R$, we have that
\begin{align}
  \sum_{i=1}^N \alpha_i(f(a_i) - f(b_i)) = f(v_1) - f(v_2).\label{eq:f-sum}
\end{align}
Arbitrarily enumerate the connected components of $E_{1,2}(H)$ and consider $f : V_1 \cup V_2 \to \mathbb N$ that maps each vertex to the index of its connected component. Assuming that $(v_1, v_2)$ bridges two components of $E_{1,2}(H)$, we have that the RHS of \cref{eq:f-sum} is nonzero, but every term in the LHS of \cref{eq:f-sum} equals zero, a contradiction.\end{proof}

\subsection{Equivalence Relations and Commutativity of Cycles}

Let $H \subseteq V_1 \times V_2 \times V_3$ be a hypergraph. Although $(1,2)$-cycles are quite useful for understanding the structure of $\langle H\rangle_{\fL}$, we actually find that it is useful to study a broader family of words on $\fL$. More precisely, we define equivalence relations $\sim_1$ and $\sim_2$ on $V_3$ as follows.

\begin{definition}
Let $H \subseteq V_1 \times V_2 \times V_3$ be a hypergraph. Given $u, v \in V_3$, we say that $u \widehat{\sim}_1 v$ if there exists $a \in V_1$ and $b, c \in V_2$ such that $(a, b, u), (a, c, v) \in H$. We let $\sim_1$ be the equivalence relation on $V_3$ induced by taking the transitive closer of $\widehat{\sim}_1$. That is, $u \sim_1 v$ if and only if there exists a sequence $a_1, \hdots, a_\ell \in V_3$ with $a_1 = u$, $a_\ell = v$ and $a_i \widehat{\sim}_1 a_{i+1}$ for all $i \in [\ell-1]$. We define $\widehat{\sim}_2$ and $\sim_2$ in an analogous manner, where $u \widehat{\sim}_2 v$ if there exists $a, b \in V_1$ and $c \in V_2$ such that $(a, c, u), (b, c, v) \in V_3$.
\end{definition}

When multiple hypergraphs are considered simultaenously, we let $\sim_{H,1}$ and $\sim_{H,2}$ to denote the equivalence relations induced by $H$. However, as we now prove, these equivalence relations depend only on the Abelian closure of $H$.

\begin{proposition}\label{prop:Abelian-equiv-relation}
Assume that $H \subseteq H' \subseteq V_1 \times V_2 \times V_3$ and $\langle H\rangle_{\mathbb Z} = \langle H'\rangle_{\mathbb Z}$. Then, for $i \in \{1,2\}$, $\sim_{H,i}$ is the same equivalence relation on $V_3$ as $\sim_{H',i}$.
\end{proposition}
 \begin{proof}
We prove that $\sim_{H,1}$ and $\sim_{H',1}$ are the same. Proving that $\sim_{H,2}$ and $\sim_{H',2}$ are the same can be done by an analogous argument. Since $H'$ has every edge of $H$, $\sim_{H',1}$ is a coarser equivalence relation than $\sim_{H,1}$--that is, $u \sim_{H,1} v$ implies $u \sim_{H',1} v$.

Define $E_{1,3}(H) = \{(v_1, v_3) : (v_1, v_2, v_3) \in H\}$, and define $E_{1,3}(H')$ analogously. Observe that $u \widehat{\sim}_{H,1} v$ iff $u, v \in V_3$ are connected by a path of length $2$ in $E_{1,3}(H)$. Since $\sim_{H,1}$ is the transitive closure of $\widehat{\sim}_{H,1}$, we have that $u \sim_{H,1} v$ iff $u,v \in V_3$ are in the same connected component of $E_{1,3}(H)$. Likewise, $u \sim_{H',1} v$ iff $u, v \in V_3$ are in the same connected component of $E_{1,3}(H')$.

Now, by a straightforward adaptation of the proof of \cref{prop:same}, we have that $E_{1,3}(H)$ and $E_{1,3}(H')$ have the same connected components. Thus, $\sim_{H,1}$ and $\sim_{H',1}$ are identical.
\end{proof}

Moving forward, we only compare hypergraphs for which $\langle H\rangle_{\mathbb Z} = \langle H'\rangle_{\mathbb Z}$. Thus, by \cref{prop:Abelian-equiv-relation}, we can use common equivalence relations $\sim_1$ and $\sim_2$ to study both.

\begin{lemma}\label{lem:equiv-word}
Let $H \subseteq V_1 \times V_2 \times V_3$ be a hypergraph. For $v_1, v_2 \in V_3$, if $v_1 \sim_1 v_2$, then there exists a word $\fw \in \fG(V_2)$ such that $\fw \fg_{v_1} \fg_{v_2} \in \langle H\rangle_{\fG}.$
Likewise, if $v_1 \sim_2 v_2$, then there exists a word $\fw \in \fG(V_1)$ such that $\fw \fg_{v_1} \fg_{v_2} \in \langle H\rangle_{\fG}$.
\end{lemma}

\begin{proof}
Assume $v_1 \sim_1 v_2$. If $v_1 = v_2$, then $\fg_{v_1}\fg_{v_2} = 1 \in \langle H\rangle_{\fG}$, so we can take $\fw = 1$. Otherwise, there exists a sequence $a_1, \hdots, a_\ell \in V_3$ with $a_1 = v_1$, $a_\ell = v_2$ and $a_i \widehat{\sim}_1 a_{i+1}$ for all $i \in [\ell-1]$. In particular, for each $i \in [\ell-1]$, there exists $b_i \in V_1$ and $c_i, d_i \in V_2$ such that $(b_i, c_i, a_i), (b_i, d_i, a_{i+1}) \in H$. In particular, this means for all $i \in [\ell]$, we have that
\[
    \fg_{b_i} \fg_{c_i} \fg_{a_i} \fg_{b_i} \fg_{d_i} \fg_{a_{i+1}} = \fg_{c_i} \fg_{d_i} \fg_{a_i} \fg_{a_{i+1}}\in \langle H\rangle_{\fG}.
\]
Hence,
\begin{align*}
    \langle H\rangle_{\fG} &\ni \prod_{i=1}^{\ell-1} \fg_{c_i} \fg_{d_i} \fg_{a_i} \fg_{a_{i+1}}\\
    &= \left(\prod_{i=1}^{\ell-1} \fg_{c_i} \fg_{d_i}\right)\left(\prod_{i=1}^{\ell-1} \fg_{a_i} \fg_{a_{i+1}}\right)\\
    &= \left(\prod_{i=1}^{\ell-1} \fg_{c_i} \fg_{d_i}\right) \fg_{a_1} \fg_{a_\ell}\\
    &= \left(\prod_{i=1}^{\ell-1} \fg_{c_i} \fg_{d_i}\right) \fg_{v_1} \fg_{v_2}.
\end{align*}
To finish, we observe that $\prod_{i=1}^{\ell-1} \fg_{c_i} \fg_{d_i} \in \fG(V_2)$. The case in which $v_1 \sim_2 v_2$ follows by an analogous argument.
\end{proof}

Fix $H \subseteq V_1 \times V_2 \times V_3$ and corresponding equivalence relations $\sim_1$ and $\sim_2$ on $V_3$. Given an even integer $\ell \in \mathbb N$, we say that $\mathcal C = (v_1, \hdots, v_{\ell})$ with each $v_i \in V_3$ is an \emph{equivalence cycle} if $v_i \sim_1 v_{i+1}$ for all odd $i \in [\ell]$ and $v_i \sim_2 v_{i+1 \mod \ell}$ for all even $i \in [\ell]$. We define the word corresponding to an equivalence cycle to be $\fg_{\mathcal C} := \fg_{v_1} \cdots \fg_{v_\ell}$. We extend \cref{lem:equiv-word} to observe the following property of equivalence cycles.

\begin{proposition}\label{prop:equiv-cycle-membership}
Let $\mathcal C = (v_1, \hdots, v_{\ell})$, then there exists $\fw_1 \in \G(V_1)$ and $\fw_2 \in \G(V_2)$ such that  $\fw_1\fg_{\mathcal C} \in \langle H\rangle_{\G}$ and $\fw_2\fg_{\mathcal C} \in \langle H\rangle_{\G}.$
\end{proposition}

\begin{proof}
For $i \in [\ell]$ odd, invoke \cref{lem:equiv-word} to find $\hat{\fw}_{i} \in \fG(V_2)$ such that $\hat{\fw}_i \fg_{v_i} \fg_{v_{i+1}} \in \langle H\rangle_{\G}$. Then, observe that
\begin{align*}
\langle H\rangle_{\G} &\ni (\hat{\fw}_1 \fg_{v_1} \fg_{v_2})(\hat{\fw}_3 \fg_{v_3} \fg_{v_4}) \cdots (\hat{\fw}_{\ell-1} \fg_{v_{\ell-1}} \fg_{v_\ell})\\
&= (\hat{\fw}_1 \hat{\fw}_3 \cdots \hat{\fw}_{\ell-1}) \fg_{\mathcal C}.
\end{align*}
Thus, we can pick $\fw_2 := (\hat{\fw}_1 \hat{\fw}_3 \cdots \hat{\fw}_{\ell-1}) \in \G(V_2)$. To prove the existence of $\fw_1$, first consider the case in which $v_1$ is an isolated vertex of $H$. In that case, the equivalence classes of $\sim_1$ and $\sim_2$ are singletons, so $C = (v_1, \hdots, v_1)$. This means that $\fg_{\mathcal C} = 1$, so we can pick $\fw_{1} = 1$.

Now assume there exists an edge $e := (a, b, v_1) \in H$. For $i \in [\ell-2]$ even, invoke \cref{lem:equiv-word} to find $\hat{\fw}_{i} \in \fG(V_1)$ such that $\hat{\fw}_i \fg_{v_i} \fg_{v_{i+1}} \in \langle H\rangle_{\fG}$. Also since $v_1 \sim_2 v_{\ell}$, pick $\hat{\fw}_0 \in \fG(V_1)$ such that $\hat{\fw}_0 \fg_{v_1} \fg_{v_\ell} \in \langle H\rangle_{\fG}$.

Now observe that
\begin{align*}
\langle H\rangle_{\fG} &\ni \fg_{e} (\hat{\fw}_2 \fg_{v_2} \fg_{v_3}) \cdots (\hat{\fw}_{\ell-2} \fg_{v_{\ell-2}} \fg_{v_{\ell-1}}) \fg_{e} (\hat{\fw}_0 \fg_{v_1} \fg_{v_\ell})\\
&= (\fg_a \hat{\fw}_2 \cdots \hat{\fw}_{\ell-2} \fg_a \hat{\fw}_0) (\fg_b \fg_b) (\fg_{v_1} \fg_{v_2} \fg_{v_3} \cdots \fg_{v_{\ell-1}} \fg_{v_1} \fg_{v_1} \fg_{v_\ell})\\
&= (\fg_a \hat{\fw}_2 \cdots \hat{\fw}_{\ell-2} \fg_a \hat{\fw}_0) \fg_{\mathcal C}.
\end{align*}
Thus, we can pick $\fw_1 := \fg_a \hat{\fw}_2 \cdots \hat{\fw}_{\ell-2} \fg_a \hat{\fw}_0 \in \G(V_1)$.
\end{proof}

As an immediate corollary,  we can derive a commutator for any two equivalence cycles, where we recall that $[g, h] := g^{-1}h^{-1}g h$.

\begin{proposition}\label{prop:cycle-commute}
Let $\mathcal C_u = (u_1, \hdots, u_p)$ and $\mathcal D = (v_1, \hdots, v_q)$ be equivalence cycles. Then, $[\fg_{\mathcal C}, \fg_{\mathcal D}] \in \langle H\rangle_{\fL}$.
\end{proposition}

\begin{proof}
By \cref{prop:equiv-cycle-membership}, we can pick $\fw_1 \in \G(V_1)$ and $\fw_2 \in \G(V_2)$ such that $\fw_1\fg_{\mathcal C} \in \langle H\rangle_{\G}$ and $\fw_2\fg_{\mathcal D} \in \langle H\rangle_{\G}.$ Then, since $\fg_{\mathcal C}, \fg_{\mathcal C} \in \fL$, we have that
\begin{align*}
  \langle H\rangle_{\fG} &\ni [\fw_1 \fg_{\mathcal C}, \fw_2\fg_{\mathcal D}]\\
                         &= (\fw_1 \fg_{\mathcal C})^{-1} (\fw_2 \fg_{\mathcal D})^{-1} \fw_1 \fg_{\mathcal C} \fw_2\fg_{\mathcal D}\\
                         &= (\fw_1^{-1} \fw_1) (\fw_2^{-1} \fw_2) (\fg_{\mathcal C}^{-1} \fg_{\mathcal D}^{-1} \fg_{\mathcal C} \fg_{\mathcal D})\\
                         &= [\fg_{\mathcal C}, \fg_{\mathcal D}].
\end{align*}
Thus, $[\fg_{\mathcal C}, \fg_{\mathcal D}] \in \langle H\rangle_{\fL}$.
\end{proof}

We make crucial use of \cref{prop:cycle-commute} when studying cosets of equivalence cycle words in \cref{subsec:equivalence-coset}

\subsection{Equivalence Cycle Decomposition}\label{subsec:topology}

Fix equivalence relations $\sim_1$ and $\sim_2$ on $V_3$. For an equivalence cycle $\mathcal C = (v_1, \hdots, v_\ell)$, the goal of this section is to show that the equivalence cycle word $\fg_{\mathcal C} = \fg_{v_1} \cdots \fg_{v_\ell}$ can we written canonically as a product of \emph{elementary} equivalence cycles. To define these elementary equivalence cycles, we need to consider a certain multigraph whose edge set corresponds to $V_3$.

\begin{definition}[Multigraph $\T$]
Let $\cA := \{A_1, \hdots, A_p\}$ denote the equivalence classes of $\sim_1$ and $\cB := \{B_1, \hdots, B_q\}$ denote the equivalence classes of $\sim_2$. For $v \in V_3$, let $A_v \in \cA$ be the unique equivalence class containing $v$ and let $B_V \in \cB$ be the unique equivalence class containing $v$. Consider the multigraph $\T = (\cV, \cE)$ with $\cV := \cA \cup \cB$ as the vertex set and $\cE := \{(A_v, B_v) \mid v\in V_3\}$ as the edge set, which we identify with $V_3$.
\end{definition}

We let $\cK_1, \hdots, \cK_r \subseteq \cV$ denote the connected components of $\T$. We further let $F_{\T} \subseteq V_3$ denote an arbitrary maximal forest of $\cT$. That is, $F_{\T}$ is a spanning tree when restricted to $\cK_i$ for any $i \in [\ell]$. Without loss of generality, we also assume that $B_i \in \cK_i$ for all $i \in [r]$.

Note that trails in $\T$ starting from $B_i$ for some $i \in [r]$ and returning to itself can be viewed as equivalence cycles. As such, we can define an elementary equivalence cycle using $\T$.

\begin{definition}[Elementary equivalence cycle]
Given $v \in V_3 \setminus F_{\T}$, we define the elementary equivalence cycle $\cC_v$ as follows. Assume that $v$'s edge of $\cT$ lies in $\cK_i$ for $i \in [r]$. Let $u_1, \hdots, u_p \in \F_{\T}$ be chosen such that the sequence of edges $u_1, \hdots, u_p$ is the unique shortest path from $B_i$ to $A_v$ in $F_{\T}$. Further let $v_1, \hdots, v_q \in F_{\T}$ be chosen such that the sequence of edges $v_1, \hdots, v_q$ is the unique shortest path from $B_v$ to $B_i$ in $F_{\T}$. Then, $\cC_v := (u_1, \hdots, u_p, v, v_1, \hdots, v_q)$. Note that $p$ is always odd and $q$ is always even. We let $\fw_{v} := \fg_{\cC_{v}}$ be the word corresponding to this elementary equivalence cycle.
\end{definition}

The main result of this subsection is to show that for (almost) any elementary cycle $\cC$, its corresponding word has a \emph{unique} representation as a product of elementary equivalence cycle words. 

\begin{remark}\label{rem:topology-1}
If one views equivalence cycle words as paths on a topology dictated by the graph $\T$, then we are saying the fundamental group of $\T$ is freely generated by the elementary cycles. This is a well-known fact in combinatorial topology (e.g., \cite{sunada2012topological}). However, we prove this fact in our own notation for completeness.
\end{remark}

\begin{proposition}\label{prop:topology}
Let $\cC = (v_1, \hdots, v_N)$ be a simple cycle such that $v_1, v_N \in B_i$ for some $i \in [r].$ Then, $\fg_{\cC}$ can be expressed uniquely as the product of elements of $\{\fg_u \mid u \in V_3 \setminus F_{\T}\}.$ 
\end{proposition}

\begin{proof}
Assume without loss of generality that $v_1, v_N \in B_1$. We first show that a representation of $\fg_{\cC}$ exists. If $\fg_{\cC} = 1$, this is obvious. Otherwise, by \cref{thm:G-reduced-not-equal}, we can write $\fg_{\cC}$ as a unique reduced word $\prod_{i=1}^{N'} \fg_{v'_i}$ by canceling consecutive equal terms. However, note that removing consecutive equal terms in an equivalence cycle preserves the property of being an equivalence cycle. Further one can check that $v'_1, v'_{N'} \in B_1$ must also hold. Thus, by reducing to a minimal counterexample, we can assume without loss of generality that $v_i \neq v_{i+1}$ for all $i \in [\ell-1]$.

For all $v \in F_{\T}$, define $\fw_{v} = 1$. We then claim that
\begin{align}
  \fg_{\cC} = \prod_{i=1}^N \fw_{v_i}^{(-1)^i}\label{eq:cycle-rep}
\end{align}

To argue this, let $\hat{\cC}$ be the equivalence cycle corresponding to the RHS of \cref{eq:cycle-rep}. We can describe $\hat{\cC}$ in words as follows.

\begin{enumerate}
\item For each odd $i \in [N]$ do the following steps.
\item Start at $B_1$ and walk along $F_{\T}$ to $B_{v_i}$. Cross the edge $v_i$ to $A_{v_i}$, then walk from $A_{v_i}$ to $B_1$ along $F_{\T}$ to $B_1$.\label{item:walk-2}
\item Start at $B_1$ and walk along $F_{\T}$ to $A_{v_{i+1}}$. Cross the edge $v_{i+1}$ to $B_{v_i}$, then walk from $B_{v_{i+1}}$ to $B_1$ along $F_{\T}$ to $B_1$.\label{item:walk-3}
\end{enumerate}
To check that this description is honest, note that if $v \in F_{\T}$, the process of walking from $B_1$ back to itself along edges of $F_{\T}$ always involves seeing the same edges again in opposite order (i.e., the reduced word of the walk is trivial). Thus, the word corresponding to this equivalence (sub)cycle is precisely $1$.

Further observe that since $A_{v_i} = A_{v_{i+1}}$ for all odd $i \in [N]$, the end of \cref{item:walk-2} is perfectly the reversal of \cref{item:walk-3}. Likewise for all even $i \in [N-2]$, $B_{v_i} = B_{v_{i+1}}$, so the end of \cref{item:walk-3} is the reversal of the beginning of \cref{item:walk-2} in the next iteration. Hence, the RHS of \cref{eq:cycle-rep} is equivalent to the following simple walk.
\begin{enumerate}
\item Walk from $B_1$ to $B_{v_1}$.
\item For each odd $i \in [N]$, walk from $B_{v_i}$ to $A_{v_i} = A_{v_{i+1}}$ via $v_i$, then walk to $B_{v_{i+1}}$ via $v_{i+1}$.
\item Walk from $B_{v_\ell}$ to $B_1$.
\end{enumerate}

However, by assumption, we have that $B_1 = B_{v_1} = B_{v_\ell}$. Thus, we are precisely describing $\cC$. Hence, \cref{eq:cycle-rep} holds.

To finish, we need to show this representation is unique. By logic similar to that of \cref{thm:G-reduced-not-equal}, assume for sake of contradiction that 
\begin{align}
\prod_{i=1}^{M} \fw_{u_i}^{d_i} = 1\label{eq:prod-1}
\end{align}
for some $u_1, \hdots, u_M \in V_3 \setminus F_{\T}$ and $d_1, \hdots, d_M \in \{-1,1\}$. Futher assume $d_i \neq d_{i+1}$ implies $u_i \neq u_{i+1}$ for all $i \in [M-1]$. By \cref{thm:G-reduced-not-equal}, if we view the LHS of \cref{eq:prod-1} as a word in $\fL$, its reduced form is the trivial word. Since each $\fw_{u_i}$ can be expressed as a product of $\fg_{u_i}$ and symbols $\{\fg_{v} : v \in F_{\T}\}$, we must have there is some $i \in [M-1]$ for which the $\fg_{u_i}$ in $\fw_{u_i}^{d_i}$ cancels out with $\fg_{u_{i+1}}$ in $\fw_{u_{i+1}}^{d_{i+1}}$ in the reduction procedure. In particular, $u_i = u_{i+1}$ so $d_i = d_{i+1}$. As such, the path from $u_i$ to $u_{i+1}$ according to $\fw_{u_i}^{d_i}\fw_{u_{i+1}}^{d_{i+1}}$ either starts as $A_{u_i}$ and ends at $B_{u_{i+1}}$ or it starts at $B_{u_i}$ and ends at $A_{u_{i+1}}$. In either case, the path has an odd number of edges, so word between $\fg_{u_i}$ and $\fg_{u_{i+1}}$ could not have fully reduced, a contradiction. Thus, the representation in (\ref{eq:cycle-rep}) is unique (up to the placements of the identity element).
\end{proof}

\subsection{Equivalence Cycle Cosets}\label{subsec:equivalence-coset}

Assume $H \subsetneq H' \subseteq V_1 \times V_2 \times V_3$, $\langle H\rangle_{\mathbb Z} = \langle H'\rangle_{\mathbb Z}$ but $\langle H\rangle_{\fG} \subsetneq \langle H'\rangle_{\fG}$. Further assume that every vertex of $V_3$ is incident with some edge of $H$; otherwise, such a vertex cannot be incident with $H'$, so we can delete it from $V_3$. Thus, by \cref{prop:normal-subgroup}, we have that $\langle H\rangle_{\fL} \lhd \fL$. Thus, we can define the subgroup $\fQ := \fL / \langle H\rangle_{\fL}$. We now show that any equivalence cycle of $H$ (and thus $H'$) leads to an element of $\fQ$, but we must first define the \emph{Abelianization} of an equivalence cycle.

\begin{definition}
Let $\cC = (v_1, \hdots, v_N)$ be an equivalence cycle. We let $\fe_{\cC} \in \mathbb Z^{V_3}$ be the \emph{Abelianization} of $\cC$ which we define as follows
\begin{align}
  \fe_{\cC} := \sum_{i=1}^{N} (-1)^i \fe_{v_i}.\label{eq:feC}
\end{align}
We further define $\fq_{\cC} \in \fQ$ to be
\[
  \prod_{v \in V_3} \fw_{v}^{(\fe_{\cC})_v} + \langle H \rangle_{\fL},
\]
where $\fw_v$ are the elementary equivalence cycle words from \cref{subsec:topology}.
\end{definition}

We now show that $\fq_{\cC}$ is precisely the coset of $\langle H\rangle_{\fL}$ corresponding to $\fg_{\cC}$.

\begin{proposition}\label{prop:q-equiv}
Let $\cC = (v_1, \hdots, v_N)$ be an equivalence cycle. Then,
\[
  \fq_{\cC} = \fg_{\cC} + \langle H\rangle_{\fL}.
\]
\end{proposition}

\begin{remark}
Proposition~\ref{prop:q-equiv} can be viewed as a variant of Hurewicz theorem~\cite{Hatcher02} that the Abelianization of the fundamental group is isomorphic to the first homology group.
\end{remark}

\begin{proof}
Assume without loss of generality that $\cC$ is within connected component $\cK_1$ of $\T$. Let $(u_1, \hdots, u_p)$ be a walk within $F_{\T}$ from $B_1$ to $B_{v_1} = B_{v_N}$. Note that since $\T$ is bipartite, $p$ is even. Let $\cC'$ be the equivalence cycle $(u_1, \hdots, u_p, v_1, \hdots, v_N, u_{p}, \hdots, u_1)$. Let $\fw := \fg_{u_1} \cdots \fg_{u_p}$. Then,
\[
  \fg_{\cC'} = \fw \fg_{\cC}\fw^{-1}.
\]
Since $\langle H\rangle_{\fL}$ is a normal subgroup of $\fL$, we have that $\fg_{\cC} + \langle H\rangle_{\fL} = \fg_{\cC'} + \langle H\rangle_{\fL}$. Furthermore, by \cref{prop:topology} and \cref{eq:cycle-rep}, we have that
\[
  \fg_{\cC'} = \left(\prod_{i=1}^p \fw_{u_i}^{(-1)^i}\right)\left(\prod_{i=1}^N \fw_{v_i}^{(-1)^i}\right)\left(\prod_{i=1}^p \fw_{u_i}^{(-1)^i}\right)^{-1}
\]
Now, by \cref{prop:cycle-commute}, we have that $[\fw_u, \fw_v] \in \langle H\rangle_\fL$ for all $u,v \in V_3$. Thus, 
\[
  \fg_{\cC'} + \langle H\rangle_{\fL} = \prod_{v \in V_3} \fw_{v}^{\alpha_v - \beta_v} + \langle H\rangle_{\fL},
\]
where $\alpha_v$ is the number of times $v = c_i$ for $i \in [\ell]$ even and $\beta_v$ is the number of times $v = c_i$ for $i \in [\ell]$ odd. In particular, by \cref{eq:feC}, we have that $\fg_{\cC'} + \langle H\rangle_{\fL} = \fq_{\cC'}$. Furthermore, $\fq_{\cC'} = \fq_{\cC}$, as the contribution of $u_1, \hdots, u_p$ to $\fe_{\cC'}$ cancels out with the contribution of $u_p, \hdots, u_1$. Therefore, $\fq_{\cC} = \fg_{\cC} + \langle H\rangle_{\fL}$. 
\end{proof}

By \cref{prop:G-to-L}, we know there is a $(1,2)$-cycle $\cC = \{(a_1,b_1,c_1), \hdots, (a_\ell,b_\ell,c_\ell)\}$ for which $(a_i,b_i,c_i) \in H$ for all $i \in [\ell-1]$ and $(a_\ell, b_\ell, c_\ell) \in H' \setminus H$. Furthermore, $\fg_{a_\ell}\fg_{b_\ell} \fg_{c_\ell} \not\in \langle H\rangle_{\fG}.$ From the existence of this cycle, we have that $c_i \sim_1 c_{i+1}$ for all odd $i \in [\ell]$ and $c_i \sim_2 c_{i+1 \mod \ell}$ for all even $i \in [\ell]$. Thus, $\mathcal D := (c_1, \hdots, c_\ell)$ is an equivalence cycle for which $\fg_{\mathcal D} \in \langle H'\rangle_{\fL} \setminus \langle H\rangle_{\fL}$. We can view this equivalence cycle as a cycle in $\T$. Thus by \cref{prop:q-equiv}, we have that $\fq_{\mathcal D} \neq \langle H\rangle_{\fL}$. A contradiction will soon lead from the following proposition

\begin{proposition}\label{prop:find-quotient}
Let $\fv \in \mathbb Z^{V_3}$ be a vector for which there exists weights $\alpha_h \in \mathbb Z$ such that
\begin{align}
  \sum_{h \in H} \alpha_h \fe_{h} = \fv.\label{eq:sum-v}
\end{align}
In particular, the $V_1$ and $V_2$ components of the LHS are zero. Then,
\[
  \prod_{v \in V_3} \fw_{v}^{\fv_v} \in \langle H\rangle_{\fL}.
\]
\end{proposition}

\begin{proof}
Consider the graph $W$ on $E_{1,2}(H)$ where the edge $(a,b)$ corresponding to $h = (a,b,c) \in H$ has weight $\alpha_h$. We decompose this weighted graph into a weighted sum of $(1,2)$-cycles as follows. While $W$ still has a $(1,2)$-cycle, pick an arbitrary $(1,2)$-cycle $h_1, \hdots, h_\ell$ and give it weight $\alpha_{h_1}$. We then replace the weight of $\alpha_{h_i}$ with $\alpha_{h_i} + (-1)^i\alpha_{h_1}$. Note that by \cref{eq:sum-v}, every vertex of $W$ has total weight $0$ and deleting a weighted cycle preserves this property. Thus, when no cycles remain (so there is a tree), every leaf edge has weight $0$. In other words, the graph is empty. In other words, we found $(1,2)$-cycles $\cC_1, \hdots \cC_M$ with weights $\beta_1, \hdots, \beta_M \in \mathbb Z$ such that
\[
  \sum_{i=1}^M \beta_i \fe_{\cC_i} = \fv.
\]
Thus,
\begin{align*}
\prod_{v \in V_3} \fw_{v}^{\fv_v} + \langle H\rangle_{\fL} &= \prod_{i=1}^M \left(\prod_{v \in V_3} \fw_{v}^{(\fe_{\cC_i})_v}\right)^{\beta_i} +\langle H\rangle_{\fL}\\
&=\prod_{i=1}^M \fq_{\cC_i}^{\beta_i} +\langle H\rangle_{\fL}\\
&=\prod_{i=1}^M \fg_{\cC_i}^{\beta_i} +\langle H\rangle_{\fL}.
\end{align*}
However, since $\fg_{\cC_i} \in \langle H\rangle_\fL$ for all $i \in [M]$. Thus, $\prod_{v \in V_3} \fw_{v}^{\fv_v} \in \langle H\rangle_{\fL}.$
\end{proof}

Now, observe that $\fe_{\cD} = \sum_{i=1}^{\ell} (-1)^i\fe_{(a_i,b_i,c_i)}$ and that $\fe_{(a_i,b_i,c_i)} \in \langle H\rangle_{\mathbb Z}$ for all $i \in [\ell]$ as $\langle H'\rangle_{\mathbb Z} = \langle H\rangle_{\mathbb Z}$. Therefore, one can find coefficients $\alpha_h \in \mathbb Z$ for all $h \in H$ such that \cref{eq:sum-v} holds with $\fv = \fe_{\cD}$. Thus, 
\[
  \fq_{\cD} = \prod_{v \in V_3} \fw_{v}^{(\fe_{\cD})_v} + \langle H\rangle_{\fL} = \langle H\rangle_{\fL}. 
\]
This contradicts our assumption that $\fq_{\mathcal D} \neq \langle H\rangle_{\fL}$. Thus, we must have had that $\langle H\rangle_{\fG} = \langle H'\rangle_{\fG}$, implying that \cref{thm:3-Catalan} holds.

\section{Limitations of Abelian Covers in Arity 4}\label{sec:limit}

Although we have shown that Abelian covers and Catalan covers are equivalent for $3$-uniform $3$-partite hypergraphs, we now establish \cref{thm:4uniformIntro} that this cannot be true for $4$-uniform $4$-partite hypergraphs. We restate the result here.

\begin{theorem}\label{thm:4uniform}
    There exists a $4$-uniform hypergraph $H = (V, E)$ which admits an Abelian cover \emph{but not} a Catalan cover.
\end{theorem}

As discussed in the introduction, we consider a $4$-partite $4$-uniform hypergraph with vertex sets $V_i = \{a_i, b_i, c_i, d_i\}$ for all $i \in [4]$ and edge set
\begin{align*}
    E := \{&e_1 := (a_1,a_2,a_3,a_4),\ 
    e_2 := (b_1, b_2, b_3, b_4),\ 
    e_3 := (c_1, c_2, c_3, c_4),\ 
    e_4 := (d_1, d_2, d_3, d_4),\\
    &e_5 := (a_1, b_2, c_3, d_4),\ 
    e_6 := (b_1, a_2, d_3, c_4),\ 
    e_7 := (c_1, d_2, a_3, b_4),\ 
    e_8:= (d_1, d_2, b_3, a_4)\}.
\end{align*}
Let $V = V_1 \cup V_2 \cup V_3 \cup V_4$. t is immediate that $H = (V, E)$ has an Abelian cover as
\begin{align}
    \fe_{e_5} + \fe_{e_6} + \fe_{e_7} + \fe_{e_8} - \fe_{e_2} - \fe_{e_3} - \fe_{e_4} = \fe_{e_1}.\label{eq:abelian-4}
\end{align}
Thus, it suffices to prove that $H$ lacks a Catalan cover. Let $E' = \{e_2, e_3, e_4, e_5, e_6, e_7, e_8\}$ and $H' = (V, E')$, by \cref{eq:abelian-4}, we know that $\langle H\rangle_{\mathbb Z} = \langle H'\rangle_{\mathbb Z}$. However, we claim that $\langle H'\rangle_{\G} \subsetneq \langle H\rangle_{\G}$, where $\G = \G(V_1, V_2, V_3, V_4)$.

\begin{lemma}\label{lem:no-Catalan}
$\langle H'\rangle_{\G} \subsetneq \langle H\rangle_{\G}$. 
\end{lemma}
\begin{proof}
Assume for sake of contradiction that $\fg_{e_1} = \fg_{a_1}\fg_{a_2}\fg_{a_3}\fg_{a_4} \in \langle H'\rangle_{\fG}$. By \cref{prop:test}, this means for any group $\fN$ and any map $\psi : V \to \fN$, we have that
\begin{align}
  \langle H'\rangle_{\fN,\psi} = \langle H\rangle_{\fN,\psi}.\label{eq:psi-bad}
\end{align}
To achieve a contradiction, we construct an explicit group $\fN$ for which \cref{eq:psi-bad} cannot hold, namely the free two-step nilpotent group\footnote{Also called the free nilpotent group of class $2$.} with four generators. We state the basic facts about such group we need to prove \cref{lem:no-Catalan}, see (e.g., \cite{Tao09} or \cite[Chapter 4]{CMZ17} for a more comprehensive discussion). In particular, our group has four generators $\fu_1, \fu_2, \fu_3, \fu_4 \in \fN$ such that an arbitrary element $\fn \in \fN$ has a unique representation of the form
\begin{align}
  \fn = \prod_{i=1}^4 \fu_i^{\alpha_i} \prod_{1 \le i < j \le 4} [\fu_i, \fu_j]^{\beta_{i,j}}.\label{eq:fN}
\end{align}
for some $\alpha_1, \hdots, \alpha_4, \beta_{1,2}, \hdots, \beta_{3,4} \in \mathbb Z$. The group has the multiplication rule that
\[
  \left(\prod_{i=1}^4 \fu_i^{\alpha_i} \prod_{1 \le i < j \le 4} [\fu_i, \fu_j]^{\beta_{i,j}}\right) \cdot \left(\prod_{i=1}^4 \fu_i^{\alpha'_i} \prod_{1 \le i < j \le 4} [\fu_i, \fu_j]^{\beta'_{i,j}}\right) = \prod_{i=1}^4 \fu_i^{\alpha_i+\alpha'_i} \prod_{1 \le i < j \le 4} [\fu_i, \fu_j]^{\beta_{i,j}+\beta'_{i,j} - \alpha'_i\alpha_j}.
\]
By this multiplication rule, one can show for any $\fn, \fn', \fn'' \in \fN$, we have that $[[\fn, \fn'], \fn''] = 0$. That is, commutators commute with any group element. We also have an homomorphism $\xi : \fN \to \mathbb Z^4$ which maps
\[
\xi\left(\prod_{i=1}^4 \fu_i^{\alpha_i} \prod_{1 \le i < j \le 4} [\fu_i, \fu_j]^{\beta_{i,j}}\right) = (\alpha_1, \alpha_2, \alpha_3, \alpha_4).
\]
In particular, for any $\fn, \fn' \in \fN$, we have that $\xi(\fn\fn') = \xi(\fn'\fn)$.

Now pick $\psi : V \to \fN$ such that for all $i \in [4]$, we have that
\[
  \psi(a_i) := \fu_1, \psi(b_i) := \fu_2, \psi(c_i) := \fu_3, \psi(d_i) := \fu_4.
\]
We can also view $\psi$ as a map from $E \to \fN^4$ as follows
\begin{align*}
\psi(e_1) = (\fu_1, \fu_1, \fu_1, \fu_1), \psi(e_2) = (\fu_2, \fu_2, \fu_2, \fu_2), \psi(e_3) = (\fu_3, \fu_3, \fu_3, \fu_3), \psi(e_4) = (\fu_4, \fu_4, \fu_4, \fu_4)\\
\psi(e_5) = (\fu_1, \fu_2, \fu_3, \fu_4), \psi(e_6) = (\fu_2, \fu_1, \fu_4, \fu_3), \psi(e_7) = (\fu_3, \fu_4, \fu_1, \fu_2), \psi(e_8) = (\fu_4, \fu_3, \fu_2, \fu_1).
\end{align*}
An important observation we make is that for all $i \in [8]$, $\xi(\psi(e_i)) = \fe_{e_i}$.

To rule out \cref{eq:psi-bad}, it suffices to show that $\psi(e_1)$ is not contained in the group $\fU \subseteq \fN^4$ generated by $\{\psi(e_i) \mid i \in \{2, \hdots, 8\}\}$. Any element of $\fv \in \fU$ can be expressed (perhaps non-uniquely) in the form
\begin{align}
  \fv = \prod_{i=2}^8 \psi(e_i)^{\gamma_i} \prod_{2 \le i < j \le 8} [\psi(e_i), \psi(e_j)]^{\delta_{i,j}}\label{eq:fv}
\end{align}
for $\gamma_i, \delta_{i,j} \in \mathbb Z$. Assume for sake of contradiction that $\fv = \psi(e_1)$. Applying the map $\xi$ to both sides, and observing that $\xi([\psi(e_i), \psi(e_j)]) = 0$, we can deduce that
\[
  \fe_{e_1} = \sum_{i=2}^8 \gamma_i \fe_{e_i}.
\]
That is, $(\gamma_2, \hdots, \gamma_8)$ are the coefficients of an Abelian cover. One can show\footnote{The reason why is that $\{\fe_{e_i} : i \in \{2, \hdots, 8\}\}$ are linearly independent over $\mathbb \F_2$. This can be observed by noticing that $\fe_{e_5}, \hdots, \fe_{e_8}$each have a non-zero coordinate not used by any other vector. Thus, any non-trivial linear combination must only use $\fe_{e_2}, \fe_{e_3}, \fe_{e_4}$ which cannot occur.\label{foot}} that \cref{eq:abelian-4} is the \emph{only} possible Abelian cover of $e_1$ using $e_2, \hdots, e_8$. Thus, we must have that $\gamma_{2} = \gamma_{3} = \gamma_4 = -1$ and $\gamma_5 = \gamma_6 = \gamma_7 = \gamma_8 = 1$. Rearranging \cref{eq:fv}, we must have that
\begin{align}
  \prod_{i=1}^4 \psi(e_i) \prod_{i=5}^8 \psi(e_i)^{-1} = \prod_{2 \le i < j \le 8} [\psi(e_i), \psi(e_j)]^{\delta_{i,j}}.\label{eq:comm-eq}
\end{align}
Let $\fM \subseteq \fN$ be the Abelian subgroup of elements of the form
\[
  \fm = \prod_{1 \le i < j \le 4} [\fu_i, \fu_j]^{\beta_{i,j}}.
\]
Let $\chi : \fM \to \mathbb Z^6$ be the projection map $\chi(\fm) = (\beta_{1,2}, \beta_{1,3}, \beta_{1,4}, \beta_{2,3}, \beta_{2,4}, \beta_{3,4})$. By design, both sides of \cref{eq:comm-eq} are contained in $\fM$, so we may apply $\chi$ (as a map from $\fM^4$ to $\mathbb Z^{4 \times 6}$) to both sides to get that
\begin{align}
  \chi\left(\prod_{i=1}^4 \psi(e_i) \prod_{i=5}^8 \psi(e_i)^{-1}\right) = \sum_{2 \le i < j \le 8} \delta_{i,j}\chi([\psi(e_i), \psi(e_j)]).\label{eq:lattice}
\end{align}
With a bit of computation, we can see that
\begin{align}
  \chi\left(\prod_{i=1}^4 \psi(e_i) \prod_{i=5}^8 \psi(e_i)^{-1}\right) = \begin{pmatrix}
    1 & 1 & 1 & 1 & 1 & 1\\
    0 & 1 & 1 & 1 & 1 & 0\\
    1 & 0 & 0 & 0 & 0 & 1\\
    0 & 0 & 0 & 0 & 0 & 0.
  \end{pmatrix}\label{eq:comm-comp}
\end{align}
For example, the first row of the above matrix follows from the fact that in our 2-step nilpotent group, we have that
\[
  \fu_1\fu_2\fu_3\fu_4\fu_1^{-1} \fu_2^{-1} \fu_3^{-1} \fu_4^{-1} = \prod_{1 \le i < j\le 4} [\fu_i, \fu_j].
\]
Note that the RHS of \cref{eq:lattice} defines a lattice $\Lambda$ over $\mathbb Z^{4 \times 6}$. To get a contradiction, it suffices to show that the RHS of \cref{eq:comm-comp} is not contained in this lattice. To do this, we describe an invariant that holds for every element of $\Lambda$ but not for the RHS of \cref{eq:comm-comp}.

More precisely, consider the map $\zeta : \mathbb Z^{6 \times 4} \to \mathbb Z/4\mathbb Z$ defined by
\begin{align}
\zeta(M) := M \cdot \begin{pmatrix}1 & 1 & 0 & 3 & 0 & 2\\1 & 1 & 0 & 3 & 0 & 2\\3 & 3 & 0 & 3 & 2 & 0\\3 & 3 & 0 & 3 & 2 & 0\end{pmatrix}\mod 4,\label{eq:inv}
\end{align}
where $\cdot$ denotes dot product (i.e., sum of entry-wise products).

Observe that
\[
  \zeta\begin{pmatrix}
    1 & 1 & 1 & 1 & 1 & 1\\
    0 & 1 & 1 & 1 & 1 & 0\\
    1 & 0 & 0 & 0 & 0 & 1\\
    0 & 0 & 0 & 0 & 0 & 0.
  \end{pmatrix} \equiv 1 + 1 + 3 + 2 + 1 + 3 + 3 \equiv 2 \mod 4.
\]
However, we argume that for all $M \in \Lambda$, we have that $\zeta(M) \equiv 0 \mod 4$. To do this, it suffices to argue that $\zeta(\chi([\psi(e_i), \psi(e_j)])) \equiv 0 \mod 4$ for all $2 \le i < j \le 8$. This can be done by checking all $21$ cases, which is done in Appendix~\ref{app:A}. Thus, we contradict our earlier assumption that $\langle H'\rangle_{\G} = \langle H\rangle_{\G}$.
\end{proof}

By \cref{lem:no-Catalan}, we have that $e_1$ cannot be deduced in a Catalan cover from $e_2, \hdots, e_8$. Due to the symmetries of the hypergraph $H$, this also implies that not $e_i$ can be deduced from $E \setminus \{e_i\}$. Therefore, $H$ lacks a Catalan cover, proving \cref{thm:4uniform}.

\begin{remark}\label{rem:finite-group}
Note by the choice of $\zeta$ in \cref{eq:inv} we did not need the full power of $\fN$ to refute the existence of a Catalan cover. Rather, we could have made $\fN$ finite by imposing that all indices in \cref{eq:fN} are taken modulo $4$, making the group have size $4^{4 + \binom{4}{2}} = 2^{20}$. Call this finite group $\fN_4$. By inspection, the details of the proof remain unchanged when $\fN$ is replaced by $\fN_4$ as \cref{foot} still applies due to linear independent of $e_2, \hdots, e_8$ over $\F_2$. We make use of this observation in \cref{thm:4-uniform-Cat} to argue that a particular CSP lacking an Abelian extension has linear non-redundancy.
\end{remark}

\section{Applications to CSP Non-redundancy}\label{sec:csp}

So far, we have studied Abelian and Catalan covers intrinsically as combinatorial objects. We now apply these combinatorial results to the \emph{non-redundancy} of constraint satisfaction problems. We begin by establishing some notation.

\subsection{CSP Notation}

A CSP is specified by a relation $R \subseteq D^r$ (with $r$ denoting the arity of the CSP and $D$ the domain), along with a universe of $n$ variables $x_1, \dots x_n \in D$. The CSP consists of $m$ constraints, each of which is the result of applying the relation $R$ to an ordered subset of $r$ variables. To specify a CSP instance, we thus provide: (a) the relation $R$, (b) the number of variables $n$, and (c) a set of constraints $H \subseteq V_1 \times V_2 \times \dots \times V_r$. We use $\mathrm{CSP}(R)$ to denote the collection of all possible CSP instances using relation $R$.

\begin{definition}
    We say that a CSP instance with relation $R$, $n$ variables, and constraints $H \subseteq V_1 \times V_2 \times \dots \times V_r$ is \emph{non-redundant}, if for each constraint $h \in H$, there exists an assignment $\psi: [n] \rightarrow D$ such that:
    \begin{enumerate}
        \item $\psi(h) \notin R$.
        \item For every $h' \neq h$, $\psi(h') \in R$.
    \end{enumerate}

    We use $\NRD(R, n)$ to denote the largest number of constraints that can be in any non-redundant instance.
\end{definition}

\begin{remark}\label{rem:nrd}
    We assume without loss of generality that our CSP instances are \emph{$r$-partite}. This means we have a partition of the variables $V = V_1 \cup \cdots \cup V_r$ such that the set of constraints is a subset of $V_1 \times V_2 \times \dots \times V_r$.  Up to constant factors (depending only on $r$), this does not alter the non-redundancy (e.g., see Lemma 2.3 and Lemma 2.4 of \cite{brakensiek2025Richness}).
\end{remark}

\subsection{Extensions and Catalan Identities} 
A frequently used tool \cite{lagerkvist2020Sparsification,bessiere2020Chain,khanna2025efficient,brakensiek2025redundancy,brakensiek2025Richness} in the study of CSP non-redundancy and related questions is the use of \emph{extensions}. In particular, we say that $R \subseteq D^r$ has an extension $S \subseteq E^r$ (with $D \subseteq E$) if $R = S \cap D^r$. In this case, it is easy to prove that $\NRD(R, n) \le \NRD(S, n)$.

For example, if $E$ is a finite group and $S \subseteq E^r$ is a coset, then by \cref{fact:malt}, we have that $\NRD(R, n) \le O(|E|^2n)$. More generally, this fact holds when $S$ is what is known as a \emph{Mal'tsev relation}. That is, there exists a \emph{Mal'tsev polymorphism} $p : E^3 \to E$ such that $p(x,y,y) = p(y,y,x) = x$ for $x,y \in E$ and for all $(a_1, \hdots, a_r), (b_1, \hdots, b_r), (c_1, \hdots, c_r) \in S$, we have that
\[
  (p(a_1, b_1, c_1), p(a_2, b_2, c_2), \cdots p(a_r, b_r, c_r)) \in S.
\]
In particular, if $S$ is a Mal'tsev relation \emph{and $|E|$ is finite}, then we know that $\NRD(R, n) \le O(|E|^2n)$ by \cref{fact:malt}.

Mal'tsev extensions are a bit difficult to reason about as there is much freedom in the choice of the domain $E$ and the associated Mal'tsev polymorphism. \cite{brakensiek2025Richness} recently made significant progress in understanding arbitrary Mal'tsev extensions by showing that any relation with a Mal'tsev extension has an extension into a (possibly infinite) group. In fact, we always have an extension into the Catalan group defined in \cref{sec:Catalan}. We discuss this connection in more detail in \cref{subsec:app-1.2}. However, we first gain a deeper understanding of extensions into Abelian groups.

\subsection{Application of \texorpdfstring{\cref{thm:AbelianCoverIntro}}{Theorem 1.1} to Non-redundancy}

Assume our relation $R \subseteq D^r$ has an Abelian extension $S \subseteq E^r$. That is, $E$ is a possibly infinite group. By previous works on CSP non-redundancy \cite{jansen2019optimal,chen2020BestCase,khanna2025efficient}, we know that $\NRD(R, n) = O_{D,r}(n)$, where the hidden constant depends (rather nontrivially) on both the arity $r$ and the domain $D$ of the original relation $R$, but \emph{not} on the size of $E$. In particular, these bounds apply when $E$ is finite. Using \cref{thm:AbelianCoverIntro}, we can completely eliminate the dependence on $D$ as well as make the dependence on $r$ logarithmic. 

\begin{theorem}\label{thm:Abelian-NRD-UB}
Assume that $R \subseteq D^r$ has an extension into an Abelian group, then $\NRD(R, n) = O(n \log r)$.
\end{theorem}

\begin{proof}
Let $V = V_1 \cup \cdots \cup V_r$ be our $r$-partite vertex set with $|V| = n$. Let $H \subseteq V_1 \times \cdots \times V_r$. We claim that if $H$ has an Abelian cover, then $H$ is a redundant instance of $\CSP(R)$. Then, by invoking \cref{thm:AbelianCoverIntro}, we can deduce that any non-redundant instance of $\CSP(R)$ must have $O(n\log r)$ hyperedges.

Assuming that $H$ has an Abelian cover, we may pick $h_0 \in H$ such that $\langle H \setminus \{h_0\}\rangle_{\mathbb Z} = \langle H\rangle_{\mathbb Z}$. In particular this means that $\fe_{h_0} \in \langle H \setminus \{h_0\}\rangle_{\mathbb Z}$. Thus, there exists $h_1, \hdots, h_\ell \in H \setminus \{h_0\}$ and coefficients $a_1, \hdots, a_\ell \in \mathbb Z$ such that
\[
    \fe_{h_0} = \sum_{i=1}^{\ell} a_i \fe_{h_i}.
\]
We use this fact to prove that $H$ is a redundant instance of $\CSP(R)$. Assume that $h_i = (v_{i,1}, \hdots, v_{i,r})$ for all $i \in \{0\}\cup [\ell]$. In particular, consider any assignment $\psi : V \to D$ such that\footnote{Given $h = (v_1, \hdots, v_r)\in H$, we let $\psi(h) := (\psi(v_1), \hdots, \psi(v_r)).$} $\psi(h_i) \in R$ for all $i \in [\ell]$. Also define vectors $w_1, \hdots, w_r \in \mathbb Z^V$ such that for all $v \in V$, we have that
\[
    w_{i,v} := \begin{cases}
    \varphi(v) & e \in V_i\\
    0 & \text{otherwise}.
    \end{cases}
\]
Let $\langle \cdot, \cdot\rangle : \Z^2 \to \Z$ be the standard bilinear form (a.k.a. ``inner product''). The key calculation is that
\begin{align*}
    R \ni \sum_{i=1}^{\ell} a_i \psi(h_i) &= \sum_{i=1}^{\ell} (a_i \varphi(v_{i,1}), a_i \varphi(v_{i,2}), \hdots, a_i \varphi(v_{i,r}))\\
    &= \sum_{i=1}^{\ell} (a_i \langle \fe_{v_{i,1}}, w_1\rangle, a_i \langle \fe_{v_{i,2}}, w_2\rangle, \hdots, a_i \langle \fe_{v_{i,r}}, w_r\rangle)\\
    &= \left(\left\langle \sum_{i=1}^{\ell} a_i \fe_{v_{i,1}}, w_1\right\rangle, \left\langle \sum_{i=1}^{\ell} a_i \fe_{v_{i,2}}, w_2\right\rangle, \hdots, \left\langle \sum_{i=1}^{\ell} a_i \fe_{v_{i,r}}, w_r\right\rangle\right)\\
    &= \left(\left\langle \fe_{v_{0,1}}, w_1\right\rangle, \left\langle \fe_{v_{0,2}}, w_2\right\rangle, \hdots, \left\langle \fe_{v_{0,r}}, w_r\right\rangle\right)\\
    &= (\varphi(v_{0,1}), \varphi(v_{0,2}), \hdots, \varphi(v_{0,r}))\\
    &= \varphi(h_0).
\end{align*}
In other words, any assignment which satisfies $h_1, \hdots, h_\ell$ also satisfies $h_0$. Thus, $H$ is a redundant instance of $\CSP(R)$, as desired.\end{proof}

\begin{remark}
If one drops the $r$-partite assumption on the instance (see \cref{rem:nrd}), we get a weaker bound of $\NRD(R, n) = O_r(n)$. However, our result is the first bound to not depend on the domain size $D$.
\end{remark}

\subsection{Application of \texorpdfstring{\cref{thm:CatalanArity3}}{Theorem 1.2} to Non-redundancy}\label{subsec:app-1.2}

The main implication of \cref{thm:CatalanArity3} is that Mal'tsev extensions imply Abelian extensions for arity-$3$ predicates.

\begin{theorem}\label{thm:3-Malt-Abelian}
Let $R \subseteq D^3$ be a relation with a Mal'tsev extension, then $R$ has an Abelian extension.
\end{theorem}

The case of binary (arity $2$) relations is well-known (e.g., \cite{bessiere2020Chain}) by much simpler means. As an immediate corollary, we obtain the following (compare \cref{fact:AbelianCatalanIntro}).
\begin{corollary}
Let $R \subseteq D^3$ be a relation with an infinite Mal'tsev extension, then $\NRD(R, n) = O(n)$.
\end{corollary}
\begin{proof}
Apply \cref{thm:3-Malt-Abelian} to get an Abelian group extension, then apply \cref{thm:Abelian-NRD-UB} to bound $\NRD(R, n)$, noting that $r=3$ is a constant.
\end{proof}

To prove \cref{thm:3-Malt-Abelian}, we need to introduce the \emph{Catalan polymorphisms} introduced by \cite{brakensiek2025Richness}. Given a relation $R \subseteq D^r$, a polymorphism $p : D^L \to D$ is an operator such that for any tuples $t^{(1)}, \hdots, t^{(L)} \in R$, we have that
\[
  (p(t^{(1)}_1, \hdots, t^{(L)}_1), p(t^{(1)}_2, \hdots, t^{(L)}_2), \hdots, p(t^{(1)}_r, \hdots, t^{(L)}_r)) \in R.
\]
Informally, $R$ is closed under the coordinate-wise operation $p$. We let $\Pol(R)$ denote the set of polymorphisms of $R$. We can now define Catalan polymorphisms.

\begin{definition}[\cite{brakensiek2025Richness}]\label{def:Catalan}
We say that operators $\{\varphi_i : D^i\to D \mid i \in \mathbb N_{\text{odd}}\} \subseteq \Pol(R)$ are \emph{Catalan polymorphisms} if the following rules are true.
\begin{align*}
\forall x \in D, \varphi_1(x) &= x.\\
\forall x_1, \hdots, x_{k+2} \in D, \varphi_{k+2}(x_1, \hdots, x_{k+2}) &= \varphi_k(x_1, \hdots, x_{j-1},x_{j+2}, \hdots, x_{k+2})\text{ if }x_j = x_{j+1}.
\end{align*}
\end{definition}

The key result of \cite{brakensiek2025Richness} is that any Mal'tsev relation has a set of Catalan polymorphisms.

\begin{theorem}[\cite{brakensiek2025Richness}]\label{thm:maltsev-Catalan}
Let $R \subseteq D^r$ be a relation that admits a Mal'tsev polymorphism $p$. Then, there are Catalan polymorphisms $\varphi_1, \varphi_3, \hdots \in \Pol(R)$ such that $\varphi_3 = p$. 
\end{theorem}

To connect \cref{thm:maltsev-Catalan} with \cref{thm:3-Malt-Abelian}, we need to show how the Catalan polymorphisms can take advantage of Catalan covers. This idea is implicit in \cite{brakensiek2025Richness} in their question to show that Mal'tsev extensions imply group extensions, but we include a proof here for completeness.

\begin{lemma}\label{lemma:R-NRD-Catalan-NRD}
Let $R \subseteq D^r$ be a Mal'tsev relation. Let $H \subseteq V_1 \times \cdots V_r$ be an $r$-partite hypergraph. If $H$ has a Catalan cover, then $H$ is redundant instance of $\CSP(R)$.
\end{lemma}

\begin{proof}
By \cref{thm:maltsev-Catalan}, we have that there are Catalan polymorphisms $\varphi_1, \varphi_3, \varphi_5 \hdots \in \Pol(R)$. Assume that $H$ has a Catalan cover. Thus, there exists $h_0 \in H$ such that $\fg_{h_0} \in \langle H \setminus \{h_0\}\rangle_{\G}$. By definition of $\langle H \setminus \{h_0\}\rangle_{\G}$ and the fact that $\fg_h^2 = 1$ for all $h \in H$, we have that there exists $h_1, \hdots, h_\ell \in H$ (possibly with repetition) such that $\fg_{h_0} = \fg_{h_1} \cdots \fg_{h_\ell}$. Applying the homomorphism $\pi_1$ from \cref{subsec:parity}, we have that
\[
  1 = \pi_{1}(\fg_{h_0}) = \pi_1(\fg_{h_1}) \cdots \pi_1(\fg_{h_\ell}) = \ell \mod 2.
\]
Thus, $\ell$ is odd. For all $i \in [r]$, define the word $\fw^{(i)}$ to be
\[
  \fw^{(i)} = \prod_{j=1}^{\ell} \fg_{h_{j,i}}.
\]
By the commutativity relations of $\G$, we have that
\[
  \fw^{(1)} \fw^{(2)} \cdots \fw^{(r)} = \fg_{h_1} \cdots \fg_{h_\ell} = \fg_{h_0}.
\]
Note that $\fg_{h_0} = (\fg_{h_{0,1}}, \hdots, \fg_{h_{0,r}})$ is a reduced word in the sense of \cref{def:reduced-word-G}. Furthermore, the word $\fw^{(1)} \fw^{(2)} \cdots \fw^{(r)}$ satisfies condition (2) of \cref{def:reduced-word-G}. Thus, by \cref{rem:naive-reduction} and \cref{thm:G-reduced-not-equal} the \emph{naive reduction} of $\fw^{(1)} \fw^{(2)} \cdots \fw^{(r)}$ equals $\fg_{h_0}$ (as words). More precisely, fix $i \in [r]$ and run the following procedure on $\fw^{(i)}$: if two consecutive symbols of $\fw^{(i)}$ are equal, delete them both; repeat until termination. Since the resulting word is reduced, by \cref{thm:G-reduced-not-equal}, this procedure must terminate with a single symbol equal to $\fg_{h_{0,i}}$.

The key observation is that the rules of this naive reduction are \emph{identical} to the cancellation conditions of the Catalan polymorphisms in \cref{def:Catalan}! In other words, we have proved that $\fw^{(1)} \fw^{(2)} \cdots \fw^{(r)} = \fg_{h_0}$ (in $\G$) implies that for any map $f : V \to D$ we have that %
\begin{align}
\varphi_{\ell}(f(h_{1,i}), f(h_{2,i}), \hdots, f(h_{\ell, i})) &= f(h_{0,i})\text{ for all $i \in [r]$.}\label{eq:cat-h}
\end{align}
Let $f : V \to D$ be any assignment such that $f(h) \in R$ for all $h \in H \setminus \{h_0\}$. Since $h_1, \hdots, h_\ell \in H \setminus \{h_0\}$, we have that $f(h_{1}), \hdots, f(h_\ell) \in R$. Since $\varphi_\ell \in \Pol(R)$, this means that $\varphi_\ell(f(h_1), \hdots, f(h_{\ell, r})) \in R$. By \cref{eq:cat-h}, we thus have that $f(h_0) \in R$. That is, any satisfying assignment to $H \setminus \{h_0\}$ also satisfies $H$. Thus, $H$ is a redundant instance of $\CSP(R)$, as desired.
\end{proof}

We can now prove \cref{thm:3-Malt-Abelian}.

\begin{proof}[Proof of \cref{thm:3-Malt-Abelian}]
Let $R \subseteq D^3$ be our relation, and let $S \subseteq E^3$ be a possibly infinite Mal'tsev extension. Let $F := \mathbb Z^D$ we identify each $d \in D$ with the indicator $\fe_D \in F$. Let $T \subseteq (Z^{D})^3$ be the smallest coset containing $R$. We claim that $T \cap D^3 = R$; thus proving that $R$ has an Abelian extension.

Assume $T \cap D^3 \neq R$. That is, there is some $t \in T \cap D^3 \setminus R$. In other words, there exists a (finite) $k$ and tuples $t_1, \hdots, t_k \in R$ as well as coefficients $\alpha_1, \hdots, \alpha_k \in \mathbb Z$ such that
\[
  \fe_t = \sum_{i=1}^k \alpha_i\fe_{t_i}.
\]
Viewing $\{t, t_1, t_2, \hdots, t_k\}$ as a hypergraph whose vertex set is $3$ disjoint copies of $D$, we can apply \cref{thm:CatalanArity3} to deduce that $\{t, t_1, t_2, \hdots, t_k\}$ is a Catalan cover. Thus, by the proof of \cref{lemma:R-NRD-Catalan-NRD}, there exists a large odd integer $\ell$ and a Catalan polymorphism $\varphi_\ell \in \Pol(S)$ as well as indices $i_1, \hdots, i_\ell \in \{1, \hdots, k\}$ such that 
\[
  \varphi_\ell(t_{i_1}, \hdots, t_{i_\ell}) = t.
\]
This means that $t \in S \cap D^3$, contradicting the fact that $S \cap D^3 = R$. Therefore, $T \subseteq F^3$ is indeed an Abelian extension of $R$.
\end{proof}

\begin{remark}\label{rem:finite}
Using known techniques \cite{chen2020BestCase,khanna2025efficient}, this Abelian extension can be made finite. Thus, \cref{cor:arity3Maltsev} holds via application of \cref{thm:3-Malt-Abelian} and \cref{thm:Abelian-NRD-UB}.
\end{remark}

\subsection{Application of \texorpdfstring{\cref{thm:4uniformIntro}}{Theorem 1.3} to Non-redundancy}

Our main application of \cref{thm:4uniformIntro} is as follows.

\begin{theorem}\label{thm:4-uniform-Cat}
There exists $R \subseteq D^4$ which has a Mal'tsev extension (in fact, a finite group extension) but not an Abelian extension.
\end{theorem}

This improves on a result of \cite{brakensiek2025Richness} which proves an analogous fact for an arity $6$ relation.

\begin{proof}
Let $D = \{1,2,3,4\}$ and let\footnote{We omit commas and parentheses for succinctness.} $R = \{2222,3333,4444,1234,2143,3412,4321\}$. It is immediate that $R$ lacks an Abelian extension because\footnote{Here, we use $\fe_{abcd}$ as shorthand for the ordered tuple $(\fe_a, \fe_b, \fe_c, \fe_d)$. Likewise, we soon use $\fu_{abcd}$ as shorthand for $(\fu_a, \fu_b, \fu_c, \fu_d)$.}
\[
  \fe_{1111} = \fe_{1234} + \fe_{2143} + \fe_{3412} + \fe_{4321} - \fe_{2222} - \fe_{3333} - \fe_{4444}.
\]
That is, any $S \subseteq E^4$ with $R \subseteq S \cap D^4$ must have that $1111 \in S \cap D^4$.

Recall from the proof of \cref{thm:4uniformIntro} and \cref{rem:finite} the finite group $\fN_4$ with generators $\fu_1, \fu_2, \fu_3, \fu_4$, where each element is of the form
\begin{align*}
  \fn = \prod_{i=1}^4 \fu_i^{\alpha_i \mod 4} \prod_{1 \le i < j \le 4} [\fu_i, \fu_j]^{\beta_{i,j} \mod 4}.
\end{align*}
for $\alpha_1, \hdots, \alpha_4, \beta_{1,2}, \hdots, \beta_{3,4} \in \mathbb Z/4\mathbb Z$. Let $E = \fN_4$ and let $S \subseteq E^4$ be the subgroup generated by
\begin{align*}
\fu_R := \{&\fu_{2222}, \fu_{3333}, \fu_{4444}, \fu_{1234}, \fu_{2143}, \fu_{3412}, \fu_{4321}\}.
\end{align*}
Assume some tuple $abcd \in D^4$ satisfies $\fu_{abcd} \in S$. If $abcd = 1111$, then we already ruled out the possibility in the proof of \cref{thm:4uniformIntro} and \cref{rem:finite}. We now also show that any $abcd \in D^4 \setminus R \setminus \{1111\}$ is also impossible. Recall from \cref{thm:4uniformIntro} the Abelian homomorphism
\[
    \xi\left(\prod_{i=1}^4 \fu_i^{\alpha_i \mod 4} \prod_{1 \le i < j \le 4} [\fu_i, \fu_j]^{\beta_{i,j} \mod 4}\right) = (\alpha_1 \mod 4, \hdots \alpha_4 \mod 4).
\]
By this definition, $\xi(\fu_i) = \fe_i \mod 4$ for all $i \in D$. Therefore, since $\fu_{abcd}$ is in the subgroup generated by $\fu_R$, by applying the map $\xi$, we can deduce there are coefficients, $c_{2222}, \hdots, c_{4321} \in \Z/4\Z$ such that
\begin{align}
    \fe_{abcd} \equiv \sum_{t \in R} c_t \fe_t \mod 4.\label{eq:done}
\end{align}
From this equation we derive a contradiction. Assume without loss of generality that $a \neq 1$ (since $abcd \neq 1111$). This means that $c_{1234} = 0$ since $\fe_{1234}$ is the only term on the RHS of \cref{eq:done} $1$ in the first coordinate. Now among the remaining tuples of $R$, the second coordinate has the values $1$ and $2$ each appear once in $2143$ and $2222$, respectively. Since $b$ cannot equal both $1$ and $2$, at least one of $c_{2143}$ or $c_{2222}$ equals zero. Likewise, in the third coordinate $1$ and $3$ each appear once in $3412$ and $3333$, respectively. Thus, at least one of $c_{3412}$ or $c_{3333}$ equals zero. Finally, in the fourth coordinate $1$ and $4$ each appear once in $4321$ and $4444$, respectively, so at least one of $c_{4321}$ or $c_{3333}$ equals zero.

In particular at most three coefficients in the RHS of \cref{eq:done} are nonzero. In particular, these nonzero coefficients must have supported within on of $8$ sets of tuples
\begin{align*}
\{2143, 3412, 4321\}, \{2143, 3412, 4444\}, \{2143, 3333, 4321\}, \{2143, 3333, 4444\},\\
\{2222, 3412, 4321\}, \{2222, 3412, 4444\}, \{2222, 3333, 4321\}, \{2222, 3333, 4444\}.
\end{align*}
In all $8$ cases, note that the first coordinates of all three tuples are distinct. Thus, in order for \cref{eq:done} to hold at least two of the three tuples in the set must have a coefficient of zero. Thus, the only valid solutions to \cref{eq:done} have $abcd \in R$, a contradiction! Thus, $S \subseteq E^4$ is indeed a Mal'tsev extension of $R$.
\end{proof}

\section{Conclusion}

In this paper, inspired by recent progress in the area of CSP non-redundancy, we introduced the notions of Abelian and Catalan covers of hypergraphs. By using a number of tools from abstract algebra, lattice theory, and algebraic topology, we proved a number new properties about Abelian and Catalan covers.

We leave \cref{conj:linear} as our primary open question. In particular, do there exist (say) $4$-uniform hypergraphs with $n$ vertices and $\omega(n)$ hyperedges which lack a Catalan cover? By resolving \cref{conj:linear}, we can show CSP predicate with Mal'tsev extensions always have linear non-redundancy, which could be viewed as essential progress toward classifying which CSP predicates have linear non-redundancy.

\bibliographystyle{alphaurl}
\bibliography{references}

\appendix

\allowdisplaybreaks

\section{Analysis of the Lattice Invariant}\label{app:A}

From the proof of \cref{lem:no-Catalan}, we now prove that $\zeta(\chi([\psi(e_i), \psi(e_j)])) \equiv 0 \mod 4$ for all $2 \le i < j \le 8$. Recall from \cref{eq:inv} that

\begin{align*}
\zeta(M) := M \cdot \begin{pmatrix}1 & 1 & 0 & 3 & 0 & 2\\1 & 1 & 0 & 3 & 0 & 2\\3 & 3 & 0 & 3 & 2 & 0\\3 & 3 & 0 & 3 & 2 & 0\end{pmatrix}\mod 4.
\end{align*}
And also recall that
\begin{align*}
\psi(e_1) = (\fu_1, \fu_1, \fu_1, \fu_1), \psi(e_2) = (\fu_2, \fu_2, \fu_2, \fu_2), \psi(e_3) = (\fu_3, \fu_3, \fu_3, \fu_3), \psi(e_4) = (\fu_4, \fu_4, \fu_4, \fu_4)\\
\psi(e_5) = (\fu_1, \fu_2, \fu_3, \fu_4), \psi(e_6) = (\fu_2, \fu_1, \fu_4, \fu_3), \psi(e_7) = (\fu_3, \fu_4, \fu_1, \fu_2), \psi(e_8) = (\fu_4, \fu_3, \fu_2, \fu_1).
\end{align*}

We demonstrate this computation in batches: the case $2 \le i < j \le 4$, the case $2 \le i \le 4 < j \le 8$, and the case $5 \le i < j \le 8$.

\subsection{\texorpdfstring{$2 \le i < j \le 4$}{2 <= i < j <= 4}}

For the first batch, observe that $[\psi(e_2), \psi(e_3)] = ([\fu_2, \fu_3], [\fu_2, \fu_3], [\fu_2, \fu_3], [\fu_2, \fu_3])$. Thus,
\[
\zeta(\chi([\psi(e_2), \psi(e_3)])) = \zeta\begin{pmatrix}
0 & 0 & 0 & 1 & 0 & 0\\
0 & 0 & 0 & 1 & 0 & 0\\
0 & 0 & 0 & 1 & 0 & 0\\
0 & 0 & 0 & 1 & 0 & 0
\end{pmatrix} \equiv 3 + 3 + 3 + 3 \equiv 0 \mod 4.
\]
Likewise,
\[
\zeta(\chi([\psi(e_2), \psi(e_4)])) = \zeta\begin{pmatrix}
0 & 0 & 0 & 0 & 1 & 0\\
0 & 0 & 0 & 0 & 1 & 0\\
0 & 0 & 0 & 0 & 1 & 0\\
0 & 0 & 0 & 0 & 1 & 0
\end{pmatrix} \equiv 0 + 0 + 2 + 2 \equiv 0 \mod 4.
\]
and
\[
\zeta(\chi([\psi(e_3), \psi(e_4)])) = \zeta\begin{pmatrix}
0 & 0 & 0 & 0 & 0 & 1\\
0 & 0 & 0 & 0 & 0 & 1\\
0 & 0 & 0 & 0 & 0 & 1\\
0 & 0 & 0 & 0 & 0 & 1
\end{pmatrix} \equiv 2 + 2 + 0 + 0 \equiv 0 \mod 4.
\]

\subsection{\texorpdfstring{$2 \le i \le 4 < j \le 8$}{2 <= i <= 4 < j <= 8}}
For our next batch, observe that 
\[
[\psi(e_2), \psi(e_5)] = ([\fu_2, \fu_1], [\fu_2, \fu_2], [\fu_2, \fu_3], [\fu_2, \fu_4]) = ([\fu_1, \fu_2]^{-1}, 1, [\fu_2, \fu_3], [\fu_2, \fu_4]).
\]
Therefore,
\[
\zeta(\chi([\psi(e_2), \psi(e_5)])) = \zeta\begin{pmatrix}
-1 & 0 & 0 & 0 & 0 & 0\\
0 & 0 & 0 & 0 & 0 & 0\\
0 & 0 & 0 & 1 & 0 & 0\\
0 & 0 & 0 & 0 & 1 & 0
\end{pmatrix} \equiv -1 + 3 + 2 \equiv 0 \mod 4.
\]
By similar logic, we can deduce that
\begin{align*}
\zeta(\chi([\psi(e_2), \psi(e_6)])) &= \zeta\begin{pmatrix}
0 & 0 & 0 & 0 & 0 & 0\\
-1 & 0 & 0 & 0 & 0 & 0\\
0 & 0 & 0 & 0 & 1 & 0\\
0 & 0 & 0 & 1 & 0 & 0
\end{pmatrix} \equiv -1 + 3 + 2 \equiv 0 \mod 4,\\
\zeta(\chi([\psi(e_2), \psi(e_7)])) &= \zeta\begin{pmatrix}
0 & 0 & 0 & 1 & 0 & 0\\
0 & 0 & 0 & 0 & 1 & 0\\
-1 & 0 & 0 & 0 & 0 & 0\\
0 & 0 & 0 & 0 & 0 & 0
\end{pmatrix} \equiv 3 + 0 - 3 \equiv 0 \mod 4,\\
\zeta(\chi([\psi(e_2), \psi(e_8)])) &= \zeta\begin{pmatrix}
0 & 0 & 0 & 0 & 1 & 0\\
0 & 0 & 0 & 1 & 0 & 0\\
0 & 0 & 0 & 0 & 0 & 0\\
-1 & 0 & 0 & 0 & 0 & 0
\end{pmatrix} \equiv 3 + 0 - 3 \equiv 0 \mod 4.
\end{align*}

For $i = 3$, we get
\begin{align*}
\zeta(\chi([\psi(e_3), \psi(e_5)])) &= \zeta\begin{pmatrix}
0 & -1 & 0 & 0 & 0 & 0\\
0 & 0 & 0 & -1 & 0 & 0\\
0 & 0 & 0 & 0 & 0 & 0\\
0 & 0 & 0 & 0 & 0 & 1
\end{pmatrix} \equiv -1 - 3 + 0\equiv 0 \mod 4,\\
\zeta(\chi([\psi(e_3), \psi(e_6)])) &= \zeta\begin{pmatrix}
0 & 0 & 0 & -1 & 0 & 0\\
0 & -1 & 0 & 0 & 0 & 0\\
0 & 0 & 0 & 0 & 0 & 1\\
0 & 0 & 0 & 0 & 0 & 0
\end{pmatrix} \equiv -3 -1 + 0 \equiv 0 \mod 4,\\
\zeta(\chi([\psi(e_3), \psi(e_7)])) &= \zeta\begin{pmatrix}
0 & 0 & 0 & 0 & 0 & 0\\
0 & 0 & 0 & 0 & 0 & 1\\
0 & -1 & 0 & 0 & 0 & 0\\
0 & 0 & 0 & -1 & 0 & 0
\end{pmatrix} \equiv 2-3-3\equiv 0 \mod 4,\\
\zeta(\chi([\psi(e_3), \psi(e_8)])) &= \zeta\begin{pmatrix}
0 & 0 & 0 & 0 & 0 & 1\\
0 & 0 & 0 & 0 & 0 & 0\\
0 & 0 & 0 & -1 & 0 & 0\\
0 & -1 & 0 & 0 & 0 & 0
\end{pmatrix} \equiv 2-3-3 \equiv 0 \mod 4.
\end{align*}

For $i = 4$, we get
\begin{align*}
\zeta(\chi([\psi(e_4), \psi(e_5)])) &= \zeta\begin{pmatrix}
0 & 0 & -1 & 0 & 0 & 0\\
0 & 0 & 0 & 0 & -1 & 0\\
0 & 0 & 0 & 0 & 0 & -1\\
0 & 0 & 0 & 0 & 0 & 0
\end{pmatrix} \equiv -0-0-0\equiv 0 \mod 4,\\
\zeta(\chi([\psi(e_4), \psi(e_6)])) &= \zeta\begin{pmatrix}
0 & 0 & 0  & 0 & -1 & 0\\
0 & 0 & -1 & 0 & 0 & 0\\
0 & 0 & 0 & 0 & 0 & 0\\
0 & 0 & 0 & 0 & 0 & -1
\end{pmatrix} \equiv -0-0-0\equiv 0 \mod 4,\\
\zeta(\chi([\psi(e_4), \psi(e_7)])) &= \zeta\begin{pmatrix}
0 & 0 & 0 & 0 & 0 & -1\\
0 & 0 & 0 & 0 & 0 & 0\\
0 & 0 & -1 & 0 & 0 & 0\\
0 & 0 & 0 & 0 & -1 & 0
\end{pmatrix} \equiv -2-0-2\equiv 0 \mod 4,\\
\zeta(\chi([\psi(e_4), \psi(e_8)])) &= \zeta\begin{pmatrix}
0 & 0 & 0 & 0 & 0 & 0\\
0 & 0 & 0 & 0 & 0 & -1\\
0 & 0 & 0 & 0 & -1 & 0\\
0 & 0 & -1 & 0 & 0 & 0
\end{pmatrix} \equiv -2-2-0 \equiv 0 \mod 4.
\end{align*}

\subsection{\texorpdfstring{$5 \le i < j \le 8$}{5 <= i < j <= 8}}

For our last batch, observe that  $[\psi(e_5), \psi(e_6)] = ([\fu_1, \fu_2], [\fu_2, \fu_1], [\fu_3, \fu_4], [\fu_4, \fu_3])$. Therefore,
\[
\zeta(\chi([\psi(e_5), \psi(e_6)])) = \zeta\begin{pmatrix}
1 & 0 & 0 & 0 & 0 & 0\\
-1 & 0 & 0 & 0 & 0 & 0\\
0 & 0 & 0 & 0 & 0 & 1\\
0 & 0 & 0 & 0 & 0 & -1
\end{pmatrix} \equiv 1 - 1 + 0 - 0 \equiv 0 \mod 4.
\]
Likewise, we can deduce that
\begin{align*}
\zeta(\chi([\psi(e_5), \psi(e_7)])) &= \zeta\begin{pmatrix}
0 & 1 & 0 & 0 & 0 & 0\\
0 & 0 & 0 & 0 & 1 & 0\\
0 & -1 & 0 & 0 & 0 & 0\\
0 & 0 & 0 & 0 & -1 & 0
\end{pmatrix} \equiv 1 + 0 - 3 - 2 \equiv 0 \mod 4,\\
\zeta(\chi([\psi(e_5), \psi(e_8)])) &= \zeta\begin{pmatrix}
0 & 0 & 1 & 0 & 0 & 0\\
0 & 0 & 0 & 1 & 0 & 0\\
0 & 0 & 0 & -1 & 0 & 0\\
0 & 0 & -1 & 0 & 0 & 0
\end{pmatrix} \equiv 0 + 3 - 3 + 0 \equiv 0 \mod 4,\\
\zeta(\chi([\psi(e_6), \psi(e_7)])) &= \zeta\begin{pmatrix}
0 & 0 & 0 & 1 & 0 & 0\\
0 & 0 & 1 & 0 & 0 & 0\\
0 & 0 & -1 & 0 & 0 & 0\\
0 & 0 & 0 & -1 & 0 & 0
\end{pmatrix} \equiv 3 + 0 - 0 - 3 \equiv 0 \mod 4,\\
\zeta(\chi([\psi(e_6), \psi(e_8)])) &= \zeta\begin{pmatrix}
0 & 0 & 0 & 0 & 1 & 0\\
0 & 1 & 0 & 0 & 0 & 0\\
0 & 0 & 0 & 0 & -1 & 0\\
0 & -1 & 0 & 0 & 0 & 0
\end{pmatrix} \equiv 0 + 1 - 2 - 3 \equiv 0 \mod 4,\\
\zeta(\chi([\psi(e_7), \psi(e_8)])) &= \zeta\begin{pmatrix}
0 & 0 & 0 & 0 & 0 & 1\\
0 & 0 & 0 & 0 & 0 & -1\\
1 & 0 & 0 & 0 & 0 & 0\\
-1 & 0 & 0 & 0 & 0 & 0
\end{pmatrix} \equiv 2 - 2 + 3 - 3 \equiv 0 \mod 4.
\end{align*}
This exhausts all cases, so $\zeta(M) \equiv 0 \mod 4$ for all $M \in \Lambda$.

\end{document}